\documentclass[11pt,letterpaper]{article}

\usepackage[utf8]{inputenc}

\usepackage{amsthm}

\usepackage{amsmath,amssymb,amsfonts} 
\usepackage{epsfig} \usepackage{latexsym,nicefrac,bbm}

\usepackage{color,fancybox,graphicx,subfigure}
\usepackage[top=1in, bottom=1in, left=1in, right=1in]{geometry}
\usepackage{tabularx} \usepackage{hyperref} 

\usepackage[boxruled, linesnumbered]{algorithm2e}
\usepackage{enumitem}
\usepackage{multirow}
 \usepackage{url}

 \usepackage{tcolorbox}

\renewcommand{\epsilon}{\varepsilon}

\newcommand{\nfrac}{\nicefrac}
\newcommand{\eps}{\varepsilon}

\newtheorem{theorem}{Theorem}[section]

\newtheorem{definition}{Definition}[section]
\newtheorem{lemma}[theorem]{Lemma}

\newtheorem{corollary}[theorem]{Corollary}

\newtheorem{assumption}[theorem]{Assumption}

\def\FullBox{\hbox{\vrule width 6pt height 6pt depth 0pt}}

\def\qed{\ifmmode\qquad\FullBox\else{\unskip\nobreak\hfil
\penalty50\hskip1em\null\nobreak\hfil\FullBox
\parfillskip=0pt\finalhyphendemerits=0\endgraf}\fi}

\usepackage[T1]{fontenc}
\usepackage{lmodern}

\title{\bf Re-Analyze Gauss: Bounds for Private Matrix Approximation  via Dyson Brownian Motion\footnote{This is the full version of a paper accepted to NeurIPS 2022 \url{https://openreview.net/pdf?id=Ep98SUx9gka}}}

 \author{Oren Mangoubi\\ Worcester Polytechnic Institute \and Nisheeth K. Vishnoi \\ Yale University}

\begin{document}
\date{}

\maketitle

\begin{abstract}
Given a symmetric matrix $M$ and a vector $\lambda$, we present new bounds on the Frobenius-distance utility of the Gaussian mechanism for  approximating $M$ by a matrix whose spectrum is $\lambda$, under $(\eps,\delta)$-differential privacy. Our bounds depend on both $\lambda$ and the gaps in the eigenvalues of $M$, and hold whenever the top $k+1$ eigenvalues of $M$ have sufficiently large gaps. When applied to the problems of private rank-$k$ covariance matrix approximation and subspace recovery, our bounds yield improvements over previous bounds. Our bounds are obtained by viewing the addition of Gaussian noise as a continuous-time matrix Brownian motion. This viewpoint allows us to track the evolution of eigenvalues and eigenvectors of the matrix, which are governed by  stochastic differential equations discovered by Dyson. These equations allow us to bound the utility as the square-root of a sum-of-squares of perturbations to the eigenvectors, as opposed to a sum of perturbation bounds obtained via Davis-Kahan-type theorems.

\end{abstract}

\newpage

\vspace{-5mm}
\tableofcontents
\newpage

\section{Introduction}

Given a dataset $A \in \mathbb{R}^{m \times d}$, which consists of $m$ individuals with $d$-dimensional features, 
 methods for preprocessing  or prediction from $A$ often use the covariance matrix $M:= A^\top A$  of $A$. 
In many such applications one computes  a rank-$k$ approximation to $M$, or finds a matrix {\em close} to $M$ with a specified set of eigenvalues $\lambda=(\lambda_1,\ldots,\lambda_d)$ \cite{shikhaliev2019low, hubert2004robust, shen2008sparse}.
Examples include the rank-$k$ covariance matrix approximation problem where one seeks to compute a rank-$k$ matrix that minimizes a given distance to  $M$, 
and the subspace recovery problem where the goal is to compute a rank $k$-projection matrix $H = V_k V_k^\top$, where $V_k$ is the $d \times k$ matrix whose columns are the top-$k$ eigenvectors of $M$.
These matrix approximation  problems are ubiquitous in ML and have a rich algorithmic history; see \cite{james2013introduction,vogels2019powersgd,candes2010matrix,blum2020foundations}.

In some cases, the rows of $A$ correspond to sensitive features of  individuals and the release of solutions to aforementioned matrix approximation problems may reveal their private information, e.g., as in the case of the Netflix prize problem   \cite{bennett2007netflix}. 
Differential privacy (DP) has become a popular notion to quantify the extent to which an algorithm preserves the privacy of individuals  \cite{dwork2006differential}.
Algorithms for solving low-rank matrix approximation problems have been widely studied under DP constraints   \cite{kapralov2013differentially, blum2005practical, dwork2014analyze, dwork2006calibrating}.
Notions of DP studied in the literature include $(\eps, \delta)$-DP \cite{dwork2006calibrating, hardt2012beating, hardt2013beyond, dwork2014analyze} which is the notion we study in this paper, as well as pure $(\eps, 0)$-DP \cite{dwork2006calibrating, kapralov2013differentially, amin2019differentially, leake2020polynomial}.
To define a notion of DP in problems involving covariance matrices, following \cite{blum2005practical, dwork2006calibrating}, two matrices $M=A^\top A$ and $M' = A'^\top A'$ are said to be {\em neighbors} if they arise from  $A, A'$ which differ by at most one row.
Moreover, as is oftentimes done, we assume that each row of the datasets $A, A'$ has norm at most $1$.
For any $\eps, \delta \geq 0$, a randomized mechanism $\mathcal{A}$ is $(\eps, \delta)$-differentially private  if for all neighbors $M, M' \in  \mathbb{R}^{d \times d}$, and any measurable subset $S$ of outputs of $\mathcal{A}$, we have $\mathbb{P}(\mathcal{A}(M) \in S) \leq e^\eps \mathbb{P}(\mathcal{A}(M') \in S) + \delta$.

{\bf The problem.} We consider a  class of problems where one wishes to compute an approximation to a symmetric $d \times d$ matrix under $(\eps,\delta)$-differential privacy constraints.
Specifically,  given $M = A^\top A$ for  $A \in \mathbb{R}^{m \times d}$, together with a vector $\lambda$ of target eigenvalues $\lambda_1 \geq \cdots \geq \lambda_d$,  the goal is to output a $d \times d$ matrix $\hat{H}$ with eigenvalues $\lambda$ which minimizes the Frobenius-norm distance
$\|\hat{H}- H\|_F$ under $(\eps,\delta)$-differential privacy constraints.
Here $H$ is the matrix with eigenvalues $\lambda$ and the same eigenvectors as $M$.
This class of problems includes as a special case the subspace recovery problem if we set  $\lambda_1 = \cdots = \lambda_k = 1$ and $\lambda_{k+1} = \cdots = \lambda_d = 0$.
It also includes the  rank-$k$ covariance approximation problems if we set  $\lambda_i = \sigma_i$ for $i \leq k$, where $\sigma_1 \geq \cdots \geq \sigma_d$ are the eigenvalues of $M$.
Since revealing $\sigma_i$s may violate privacy constraints,  the eigenvalues of the output matrix $\hat{H}$ should not be the same as those of $H$.

Various distance functions have been used in the literature to evaluate the utility of $(\eps,\delta)$-DP mechanisms for matrix approximation problems, including the Frobenius-norm distance $\|\hat{H}  - H\|_F$  (e.g. \cite{dwork2014analyze, amin2019differentially})  
and the  Frobenius inner product utility $\langle M, H - \hat{H}\rangle$ (e.g. \cite{chaudhuri2012near, dwork2014analyze, gilad2017smooth}).
Note that while a bound $\| H -\hat{H} \|_F \leq b$  implies an upper bound on the inner product utility of $\langle M, H -\hat{H} \rangle \leq \| M\|_F \cdot b$ (by the Cauchy-Schwarz inequality),
an upper bound on the inner product utility  does not (in general) imply any upper bound on the Frobenius-norm distance.
Moreover, the Frobenius-norm distance can be a good utility metric to use if the goal is to recover a low-rank matrix $H$ from a dataset of noisy observations (see e.g. \cite{davenport2016overview}). 
Hence, we use the Frobenius-norm distance to measure the utility of an $(\eps,\delta)$-DP mechanism.

{\bf Related work.} The problem of approximating a matrix under differential privacy constraints has been widely studied. 
In particular, prior works have provided algorithms for problems where the goal is to approximate a covariance matrix under differential privacy constraints, including rank-$k$ PCA and subspace recovery \cite{blum2005practical, kapralov2013differentially, dwork2014analyze, leake2020computability} as well as rank-$k$ covariance matrix approximation \cite{blum2005practical, dwork2014analyze, amin2019differentially}.
Another set of works has studied the problem of approximating a rectangular data matrix $A$ under DP  \cite{blum2005practical, achlioptas2007fast, hardt2012beating, hardt2013beyond}.
 We note that upper bounds on the utility of differentially-private mechanisms for rectangular matrix approximation problems can grow with the number of datapoints $m$, while those for covariance matrix approximation problems oftentimes depend only on the dimension $d$ of the covariance matrix and do not grow with $m$.
Prior works which deal with covariance matrix approximation problems such as rank-$k$ covariance matrix approximation and subspace recovery are the most relevant to our paper.
The notion of DP varies among the different works on differentially-private matrix approximation, with many of these works considering the notion $(\eps,\delta)$-DP \cite{hardt2012beating, hardt2013beyond, dwork2014analyze}, while other works focus on (pure) $(\eps, 0)$-DP \cite{kapralov2013differentially, amin2019differentially, leake2020computability}.

{\em Analysis of the Gaussian mechanism in \cite{dwork2014analyze}.} 
 \cite{dwork2014analyze} analyze a version of the Gaussian mechanism of \cite{dwork2006our}, where one perturbs the entries of $M$ by adding a symmetric matrix $E$ with i.i.d. Gaussian entries $N(0,\nfrac{\sqrt{\log(\frac{1}{\delta})}}{\eps})$, to obtain an $(\eps, \delta)$-differentially private mechanism which outputs a perturbed matrix $\hat{M} = M+E$.
One can then post-process this matrix $\hat{M}$ to obtain a rank-$k$ projection matrix which projects onto the subspace spanned by the top-$k$ eigenvectors of $\hat{M}$ (for the rank-$k$ PCA or subspace recovery problem), or a rank-$k$ matrix $\hat{H}$ with the same top-$k$ eigenvectors and eigenvalues as $\hat{M}$ (for the rank-$k$ covariance matrix approximation problem).
\cite{dwork2014analyze} consider different notions of utility in their results, including the inner product utility (for PCA), and the Frobenius-norm and spectral-norm distance distances (for low-rank approximation and subspace recovery).

In one set of results,  \cite{dwork2014analyze} give lower utility bounds of $\tilde{\Omega}(k\sqrt{d})$ w.h.p. for the rank-$k$ PCA problem with respect to the inner product utility $\langle M, H \rangle$, together with matching upper bounds provided by a post-processing of the Gaussian mechanism, where $\tilde{\Omega}$ hides polynomial factors of $\frac{1}{\eps}$ and $\log(\frac{1}{\delta})$ (their Theorems 3 and 18).
As noted by the authors, their lower bounds are tight for matrices $M$ with the ``worst-case'' spectral profile $\sigma$, but they can obtain improved upper bounds for matrices $M$  where  $\sigma_k - \sigma_{k+1} > \tilde{\Omega}(\sqrt{d})$ (Theorem 3 of \cite{dwork2014analyze}).

For the subspace recovery problem,  \cite{dwork2014analyze} obtain a Frobenius-distance bound of  $\|\hat{H}- H\|_F \leq \tilde{O}\left(\nfrac{\sqrt{kd}}{(\sigma_k - \sigma_{k+1})}\right)$ w.h.p. for a post-processing of the Gaussian mechanism whenever $\sigma_k - \sigma_{k+1} > \tilde{\Omega}(\sqrt{d})$ (implied by their Theorem 6, which is stated for the spectral norm).
And for the rank-$k$ covariance matrix approximation problem,  \cite{dwork2014analyze} show a utility bound of $\|\hat{H}- M\|_F - \|H- M\|_F \leq \tilde{O}(k \sqrt{d})$ w.h.p. for a post-processing of the Gaussian mechanism (Theorem 7 in \cite{dwork2014analyze}), and also give related bounds for the spectral norm.
While their Frobenius bound for the covariance matrix approximation problem is independent of the number of datapoints $m$, it may not be tight.
For instance, when $k=d$, one can easily obtain a better bound since, by the triangle inequality, $\|\hat{H}- M\|_F - \|H- M\|_F \leq \|\hat{H}- H\|_F = \|\hat{M}- M\|_F = \|E\|_F \leq O(d)$ w.h.p., since $\|E\|_F$ is just the norm of a vector of $d^2$ Gaussians with variance $\tilde{O}(1)$.
Moreover, the bound for the  rank-$k$ covariance approximation problem, $\|\hat{H}- H\|_F \leq \tilde{O}(k \sqrt{d})$, is also a worst-case upper bound for any spectral profile $\sigma$ as the right-hand side of the bound not depend on the eigenvalues $\sigma$.

{\em 
Thus, a question arises of whether the Frobenius-norm utility bounds for the rank-$k$ covariance matrix approximation and subspace recovery problems 
are tight for all spectral profiles $\sigma$, and whether the analysis of the Gaussian mechanism can be improved to achieve better utility bounds.
A more general question is to obtain utility bounds for the  Gaussian mechanism for the matrix approximation problems for arbitrary $\lambda$.
}

%
{\bf Our contribution.} Our main result is a new upper bound on the Frobenius-distance utility of the Gaussian mechanism for the  general matrix approximation problem for a given $M$ and $\lambda$ (Theorem \ref{thm_large_gap}). 
Our bound depends on the eigenvalues of $M$ and the entries of $\lambda$.

The novel insight is to view the perturbed matrix $M+E$  as a continuous-time symmetric matrix diffusion, where each entry of the matrix $M+E$ is the value reached by a (one-dimensional) Brownian motion after some time $T = \nfrac{\sqrt{\log(\frac{1}{\delta})}}{\eps}$. 
This matrix-valued Brownian motion, which we denote by $\Phi(t)$, induces a stochastic process on the eigenvalues $\gamma_1(t) \geq \cdots \geq \gamma_d(t)$ and corresponding eigenvectors $u_1(t), \ldots, u_d(t)$ of $\Phi(t)$ originally discovered by Dyson and now referred to as Dyson Brownian motion, with initial values $\gamma_i(0) = \sigma_i$ and $u_i(0)$ which are the eigenvalues and eigenvectors of the initial matrix $M$ \cite{dyson1962brownian}.

We then use the stochastic differential equations  \eqref{eq_DBM_eigenvalues} and \eqref{eq_DBM_eigenvectors}, which govern the evolution of the eigenvalues and eigenvectors of the Dyson Brownian motion, to track the perturbations to each eigenvector.
Roughly speaking, these equations say that, as the Dyson Brownian motion  evolves over time, every pair of eigenvalues $\gamma_i(t)$ and $\gamma_j(t)$, and corresponding eigenvectors $u_i(t)$ and $u_j(t)$, interacts with the other eigenvalue/eigenvector with the magnitude of the interaction term proportional to $ \frac{1}{\gamma_i(t) - \gamma_j(t)}$ at any given time $t$. 
This allows us to bound the perturbation of the eigenvectors at every time $t$, provided that the initial gaps in the top $k+1$ eigenvalues of the input matrix are $\geq \Omega(\sqrt{d})$ (Assumption \ref{assumption_gaps}).
Empirically, we observe that Assumption \ref{assumption_gaps} is satisfied for  covariance matrices of many real-world datasets (see Section \ref{appendix_data}), as well as on Wishart random matrices  $W = A^\top A$, where $A$ is an $m \times d$ matrix of i.i.d. Gaussian entries, for sufficiently large $m$ (see Section \ref{appendix_wishart}).
We then derive a stochastic differential equation that tracks how the utility changes as the Dyson Brownian motion evolves over time (Lemma \ref{Lemma_projection_differntial}) and integrate this differential equation over time to obtain a bound on the (expectation of) the utility $\mathbb{E}[\|\hat{H} -H \|_F]$ (Lemma \ref{Lemma_integral}) as a function of the gaps $\gamma_i(t) - \gamma_j(t)$.
%

Plugging in basic estimates (Lemma \ref{lemma_gap_concentration}) for the eigenvalue gaps $\gamma_i(t) - \gamma_j(t)$ to Lemma \ref{Lemma_integral}, we obtain a bound on the expected utility $\mathbb{E}[\|\hat{H} -H \|_F]$ (Theorem \ref{thm_large_gap}) for the different matrix approximation problems as a function of the eigenvalue gaps $\sigma_i - \sigma_j$ of the input matrix $M$.
Roughly speaking, our bound is the square-root of a sum-of-squares of the ratios, $\frac{\lambda_i - \lambda_j}{\sigma_i - \max(\sigma_j, \sigma_{k+1})}$, of eigenvalue gaps of the input and output matrices.

When applied to the  rank-$k$ covariance matrix approximation problem (Corollary \ref{cor_rank_k_covariance2}), Theorem \ref{thm_large_gap} implies a bound of $\mathbb{E}[\|\hat{H} -H \|_F] \leq \tilde{O}(\sqrt{k d})$ whenever the eigenvalues $\sigma$ of the input matrix $M$ satisfy $\sigma_k - \sigma_{k+1} \geq \Omega(\sigma_k)$ and the gaps in top $k+1$ eigenvalues satisfy $\sigma_i - \sigma_{i+1} \geq \tilde{\Omega}(\sqrt{d})$.
Thus, when $M$ satisfies the above condition on $\sigma$, our bound improves by a factor
of $\sqrt{k}$ on the (expectation of) the previous bound of \cite{dwork2014analyze}, which says that  $\|\hat{H}- M\|_F - \|H- M\|_F \leq \tilde{O}(k \sqrt{d})$ w.h.p., since by the triangle inequality $\|\hat{H}- M\|_F - \|H- M\|_F \leq \|\hat{H}- H\|_F$.
This condition on $\sigma$ is satisfied, e.g., for matrices $M$ whose eigenvalue gaps are at least as large as those of the Wishart random covariance matrices with sufficiently many datapoints $m$ (see Section \ref{sec_results} for details).
And, if $\sigma$ is such that $\sigma_i - \sigma_{i+1} \geq \Omega(\sigma_k - \sigma_{k+1})$ for $i \leq k$, Theorem \ref{thm_large_gap} implies a bound of $\mathbb{E}[\|\hat{H} -H \|_F] \leq \tilde{O}(\nfrac{\sqrt{d}}{(\sigma_k- \sigma_{k+1})})$ for the  subspace recovery
problem (Corollary \ref{cor_subspace_recovery}), improving by a factor of $\sqrt{k}$ (in expectation) on the previous bound of \cite{dwork2014analyze}, which implies that  $\|\hat{H}- M\|_F - \|H- M\|_F \leq \tilde{O}\left(\nfrac{\sqrt{kd}}{(\sigma_k - \sigma_{k+1})}\right)$ w.h.p.

\section{Results} \label{sec_results}

Our main result (Theorem \ref{thm_large_gap}) gives a new and unified upper bound on the Frobenius-norm utility of a post-processing of the Gaussian mechanism, for the general matrix approximation problem where one is given a symmetric matrix $M \in \mathbb{R}^{d \times d}$ and a vector $\lambda$ with $\lambda_1 \geq \cdots \geq \lambda_d$, and the goal is to compute a matrix $\hat{H}$ with eigenvalues $\lambda$ which minimizes the distance $\|\hat{H}- H\|_F$.
Here $H$ is the matrix with eigenvalues $\lambda$ and the same eigenvectors as $M$.
Plugging in different choices of $\lambda$ to Theorem \ref{thm_large_gap}, we obtain as corollaries new Frobenius-distance utility bounds for the rank-$k$ covariance matrix approximation problem (Corollary \ref{cor_rank_k_covariance2}) and the subspace recovery problem (Corollary \ref{cor_subspace_recovery}).
Our results rely on the following assumption about the eigenvalues of the input matrix $M$:

\begin{assumption}[($M,k,\lambda_1, \eps, \delta$) Eigenvalue gaps]\label{assumption_gaps}
The gaps in the top $k+1$ eigenvalues eigenvalues $\sigma_1 \geq \cdots \geq \sigma_d$ of the matrix $M \in \mathbb{R}^{d \times d}$ satisfy $\sigma_i - \sigma_{i+1} \geq  \frac{8\sqrt{\log(\frac{1.25}{\delta})}}{\eps} \sqrt{d} + 3\log^{\frac{1}{2}}(\lambda_1 k)$ for every $i \in [k]$.
\end{assumption}
\noindent
We observe empirically that Assumption \ref{assumption_gaps} is satisfied on a number of real-world datasets which were previously used as benchmarks in the differentially private matrix approximation  literature \cite{chaudhuri2012near, amin2019differentially} (see Section \ref{appendix_data}).
 Assumption \ref{assumption_gaps} is also  satisfied, for instance, by random Wishart matrices  $W = A^\top A$, where $A$ is an $m \times d$ matrix of i.i.d. Gaussian entries, which are a popular model for sample covariance matrices \cite{wishart1928generalised}.
This is because the minimum  gap $\sigma_i - \sigma_{i+1}$ of a Wishart matrix grows proportional to $\sqrt{m}$ with high probability; thus for large enough $m$, Assumption \ref{assumption_gaps} holds (see Section \ref{appendix_wishart} for details).
Hence, the assumption  requires that the gaps in the top $k+1$ eigenvalues of  $M$ are at least as large as the gaps in a random Wishart matrix.

\begin{theorem}[\bf Main result]\label{thm_large_gap}
Let $\eps, \delta>0$, and given a symmetric matrix $M \in \mathbb{R}^{d \times d}$ with  eigenvalues $\sigma_1 \geq \cdots \geq \sigma_d$ and corresponding orthonormal eigenvectors $v_1,\ldots, v_d$.
Let $G$ be a matrix with i.i.d. $N(0,1)$ entries, and consider the mechanism that outputs $\hat{M} = M+  \frac{\sqrt{2\log(\frac{1.25}{\delta})}}{\eps}(G +G^\top)$.
Then such a mechanism is $(\eps,\delta)$-differentially private. 
Moreover, let $\lambda_1 \geq \cdots \geq \lambda_d$ and $k \in [d]$ be any numbers such that $\lambda_i = 0$ for $i>k$, and define $\Lambda := \mathrm{diag}(\lambda_1, \ldots, \lambda_d)$ and $V = [v_1,\ldots, v_d]$, and define $\hat{\sigma}_1 \geq \cdots \geq \hat{\sigma}_d$ to be the eigenvalues of $\hat{M}$ with corresponding orthonormal eigenvectors $\hat{v}_1,\ldots, \hat{v}_d$ and $\hat{V} = [\hat{v}_1,\ldots,\hat{v}_d]$.
Then if $M$ satisfies Assumption \ref{assumption_gaps} for ($M,k,\lambda_1,\eps,\delta$), we have 
\begin{equation*}
    \mathbb{E}\left[\| \hat{V} \Lambda \hat{V}^\top -  V \Lambda V^\top \|_F^2\right]
\leq       O\left(\sum_{i=1}^{k}  \sum_{j = i+1}^d  \frac{(\lambda_i - \lambda_j)^2}{(\sigma_i-\max(\sigma_j, \sigma_{k+1}))^2}
     \right) \frac{\log(\frac{1}{\delta})}{\eps^2}.
     %
    %
\end{equation*}
\end{theorem}
\noindent
The fact that the mechanism in this theorem is $(\eps,\delta)$-differentially private follows from standard results about the Gaussian mechanism \cite{dwork2014analyze}.
Given any list of eigenvalues $\lambda$, and letting $\Lambda = \mathrm{diag}(\lambda)$, one can post-process the matrix $\hat{M}$ by computing its spectral decomposition $\hat{M} = \hat{V}\hat{\Sigma} \hat{V}^\top$ and  replacing its eigenvalues to obtain a matrix $\hat{V} \Lambda \hat{V}^\top$ with eigenvalues $\lambda$ and eigenvectors $\hat{V}$.
Since $\hat{V} \Lambda \hat{V}^\top$ is a post-processing of the Gaussian mechanism, the mechanism which outputs $\hat{V} \Lambda \hat{V}^\top$ is differentially private as well. 
 Theorem \ref{thm_large_gap} bounds the excess utility $\mathbb{E}[\|\hat{V} \Lambda \hat{V}^\top - V \Lambda V^\top\|_F^2]$  (whenever the gaps in the eigenvalues $\sigma_1 \geq \cdots \geq  \sigma_d$ of the input matrix satisfy Assumption \ref{assumption_gaps}) as a sum-of-squares of the ratio of the gaps $\lambda_i-\lambda_j$ in the given eigenvalues to the corresponding gaps $\sigma_i-\max(\sigma_j, \sigma_{k+1})$ in the eigenvalues of the input matrix (note that $\lambda_i-\lambda_j = \lambda_i-\max(\lambda_j, \lambda_{k+1})$ since $\lambda_j = 0$ for $j \geq k+1$).

While we do not know if Theorem \ref{thm_large_gap} is tight for all choices of $\lambda$ and $k$, it does give a tight bound for some problems.
Namely, when applied to the covariance matrix estimation problem, in the special case where $k=d$ Theorem \ref{thm_large_gap} implies a bound of $\mathbb{E}[\|\hat{M} - M \|_F] \leq \tilde{O}(\sqrt{kd}) = O(d)$ (see Corollary \ref{cor_rank_k_covariance2}).
Since $\hat{M} - M = \frac{\sqrt{2\log(\frac{1.25}{\delta})}}{\eps}(G +G^\top)$,
 the matrix $\hat{M} - M$ has independent Gaussian entries with mean zero and variance $\tilde{O}(1)$, and we have from concentration results for Gaussian random matrices (see e.g. Theorem 2.3.6 of \cite{tao2012topics}) that $\mathbb{E}[\|\hat{M} - M \|_F] = \tilde\Omega(d)$, implying that the bound in Theorem \ref{thm_large_gap} is tight in this case.

The proof of Theorem \ref{thm_large_gap} differs from prior works, including that of \cite{dwork2014analyze} which use Davis-Kahan-type theorems \cite{davis1970rotation} and trace inequalities, and instead relies on an interpretation of the Gaussian mechanism as a diffusion process which may be of independent interest (See Section \ref{sec_challenges} for additional comparison to previous approaches).
This connection allows us to use sophisticated tools from stochastic differential equations and random matrix theory. We present an outline of the proof in Section \ref{sec:proof}.

 \paragraph{Application to covariance matrix approximation:}
 Plugging $\lambda_i = \sigma_i$ for $i \leq k$ and $\lambda_i=0$ for $i>k$ into Theorem \ref{thm_large_gap}, and plugging in concentration bounds for the perturbation to the eigenvalues $\sigma_i$, we obtain utility bounds for covariance matrix approximation:

\begin{corollary}[\bf Rank-$k$ covariance matrix approximation]\label{cor_rank_k_covariance2}
Let $\eps, \delta>0$, and given a symmetric matrix $M \in \mathbb{R}^{d \times d}$ with eigenvalues $\sigma_1 \geq \cdots \geq \sigma_d$ and corresponding orthonormal eigenvectors $v_1,\ldots, v_d$.
Let $G$ be a matrix with i.i.d. $N(0,1)$ entries, and consider the mechanism that outputs $\hat{M} = M+  \frac{\sqrt{2\log(\frac{1.25}{\delta})}}{\eps}(G +G^\top)$.
Then such a mechanism is $(\eps,\delta)$-differentially private. 
Moreover, for any $k \in [d]$, define $\Sigma_k := \mathrm{diag}(\sigma_1, \ldots, \sigma_k, 0 \ldots, 0)$ and $V = [v_1,\ldots, v_d]$, and define $\hat{\sigma}_1 \geq \cdots \geq \hat{\sigma}_d$ to be the eigenvalues of $\hat{M}$ with corresponding orthonormal eigenvectors $\hat{v}_1,\ldots, \hat{v}_d$, and define $\hat{\Sigma}_k := \mathrm{diag}(\hat{\sigma}_1, \ldots, \hat{\sigma}_k, 0 \ldots, 0)$ and $\hat{V} := [\hat{v}_1,\ldots,\hat{v}_d]$. 
Then if $M$ satisfies Assumption \ref{assumption_gaps} for ($M,k,\sigma_1,\eps,\delta$), and defining $\sigma_{d+1}:=0$, we have
 $$ \mathbb{E}\left[\| \hat{V} \hat{\Sigma}_k \hat{V}^\top -  V \Sigma_k V^\top \|_F \right] \leq  O\left(\sqrt{kd} \times  \frac{\sigma_k}{\sigma_{k}-\sigma_{k+1}
    }\right) \frac{\log^{\frac{1}{2}}(\frac{1}{\delta})}{\eps}.$$

\end{corollary}
\noindent
The proof appears in Section \ref{sec:covariance}.
If $\sigma_k - \sigma_{k+1} = \Omega(\sigma_k)$, then Corollary \ref{cor_rank_k_covariance2} implies that 
    $$\mathbb{E}\left[\| \hat{V} \hat{\Sigma}_k \hat{V}^\top -  V \Sigma_k V^\top \|_F\right] \leq O\left( \sqrt{kd}  \frac{\log^{\frac{1}{2}}(\frac{1}{\delta})}{\eps}\right).$$
Thus, for matrices $M$ with eigenvalues satisfying Assumption \ref{assumption_gaps} and where $\sigma_k - \sigma_{k+1} = \Omega(\sigma_k)$, Corollary \ref{cor_rank_k_covariance2} improves by a factor of $\sqrt{k}$ on the bound in Theorem 7 of \cite{dwork2014analyze} which says $\| \hat{V} \hat{\Sigma}_k \hat{V}^\top - M\|_F - \|V \Sigma_k V^\top - M\|_F = \tilde{O}(k\sqrt{d})$ w.h.p..
This is because an upper bound on $\| \hat{V} \hat{\Sigma}_k \hat{V}^\top -  V \Sigma_k V^\top \|_F$ implies an upper bound on $\| \hat{V} \hat{\Sigma}_k \hat{V}^\top - M\|_F - \|V \Sigma_k V^\top - M\|_F$ by the triangle inequality.
 On the other hand, while their result does not require a bound on the gaps in the eigenvalue of $M$ and bounds their utility w.h.p., our Corollary \ref{cor_subspace_recovery} requires a bound on the gaps of the top $k+1$ eigenvalues of $M$ and bounds the expected utility $\mathbb{E}[\| \hat{V} \hat{\Sigma}_k \hat{V}^\top -  V \Sigma_k V^\top \|_F ]$.

\paragraph{Application to subspace recovery:} Plugging in $\lambda_1 = \cdots = \lambda_k =1$ and $\lambda_{k+1} = \cdots = \lambda_d = 0$, the post-processing step in Theorem \ref{thm_large_gap} outputs a projection matrix, and we obtain utility bounds for the subspace recovery problem.

\begin{corollary}[\bf Subspace recovery]\label{cor_subspace_recovery}
Let $\eps,\delta>0$, and given a symmetric matrix $M \in$ $\mathbb{R}^{d \times d}$ with eigenvalues $\sigma_1 \geq \cdots \geq \sigma_d$ and corresponding orthonormal eigenvectors $v_1,\ldots, v_d$.
Let $G$ be a matrix with i.i.d. $N(0,1)$ entries, and consider the mechanism that outputs $\hat{M} = M+  \frac{\sqrt{2\log(\frac{1.25}{\delta})}}{\eps}(G +G^\top)$.
Then such a mechanism is $(\eps,\delta)$-differentially private. 
Moreover, for any $k \in [d]$, define the $d\times k$ matrices  $V_k = [v_1,\ldots,v_k]$ and $\hat{V}_k = [\hat{v}_1,\ldots,\hat{v}_k]$, where $\hat{\sigma}_1 \geq \cdots \geq \hat{\sigma}_d$ denote the eigenvalues of $\hat{M}$ with corresponding orthonormal eigenvectors $\hat{v}_1,\ldots, \hat{v}_d$. 
Then if $M$ satisfies Assumption \ref{assumption_gaps} for ($M,k,2,\eps,\delta$), we have
$    \mathbb{E}\left[\| \hat{V}_k \hat{V}_k^\top -  V_k V_k^\top \|_F \right] \leq  O\left(\frac{\sqrt{kd}}{\sigma_{k}-\sigma_{k+1}
    }\times \frac{\log^{\frac{1}{2}}(\frac{1}{\delta})}{\eps}\right).$
Moreover, if we also have that  $\sigma_i - \sigma_{i+1} \geq \Omega(\sigma_k - \sigma_{k+1})$ for all $i \leq k$, then
$$\mathbb{E}\left[\| \hat{V}_k \hat{V}_k^\top -  V_k V_k^\top \|_F \right] \leq  O\left(\frac{\sqrt{d}}{\sigma_{k}-\sigma_{k+1}
    }\times \frac{\log^{\frac{1}{2}}(\frac{1}{\delta})}{\eps}\right).$$
\end{corollary}
\noindent
The proof appears in Section \ref{sec:cor_rank_subspace}.
For matrices $M$ satisfying Assumption \ref{assumption_gaps}, the first inequality 
of Corollary \ref{cor_subspace_recovery}  recovers (in expectation) the Frobenius-norm utility bound implied by Theorem 6 of \cite{dwork2014analyze}, which states that $\| \hat{V}_k \hat{V}_k^\top -  V_k  V_k^\top \|_F \leq O\left(\frac{\sqrt{kd}}{\sigma_k - \sigma_{k+1}} \times \frac{\log^{\frac{1}{2}}(\frac{1}{\delta})}{\eps}\right)$ w.h.p.
Moreover, for many input matrices, $M$ with spectral profiles $\sigma_1 \geq \cdots \geq \sigma_d$ satisfying Assumption \ref{assumption_gaps}, Theorem \ref{thm_large_gap} implies stronger bounds than those in \cite{dwork2014analyze} for the subspace recovery problem.
For instance, if we also have that $\sigma_i - \sigma_{i+1} \geq \Omega(\sigma_k - \sigma_{k+1})$ for all $i \leq k$, the bound given in the second inequality  of Corollary \ref{cor_subspace_recovery} improves on the bound of \cite{dwork2014analyze} by a factor of $\sqrt{k}$.
 On the other hand, while their result only requires that $\sigma_k - \sigma_{k+1} \geq \sqrt{d}$ and bounds the Frobenius distance $\|\hat{V}_k \hat{V}_k^\top -  V_k  V_k^\top \|_F$ w.h.p., our Corollary \ref{cor_subspace_recovery} requires a bound on the gaps of the top $k+1$ eigenvalues of $M$ and bounds the expected Frobenius distance $\mathbb{E}[\| \hat{V}_k \hat{V}_k^\top -  V_k  V_k^\top \|_F]$.

\section{Preliminaries}

{\bf Brownian motion and stochastic calculus.}
A Brownian motion $W(t)$ in $\mathbb{R}$ is a continuous process that has stationary
independent increments  (see e.g., \cite{morters2010brownian}).
In a  multi-dimensional Brownian motion,  each coordinate is an independent and identical Brownian motion. 
The filtration $\mathcal{F}_t$  generated by  $W(t)$ is defined as $\sigma \left(\cup_{s \leq t} \sigma(W(s))\right)$, where $\sigma(\Omega)$ is the $\sigma$-algebra generated by $\Omega$.
$W(t)$ is a martingale with respect to $\mathcal{F}_t$.

\begin{definition}[\bf It\^o Integral]
Let $W(t)$ be a Brownian motion for $t \geq 0$, let $\mathcal{F}_t$ be the filtration generated by $W(t)$, and let $z(t) :  \mathcal{F}_t \rightarrow \mathbb{R}$ be a stochastic process adapted to $\mathcal{F}_t$.
The It\^o integral is defined as 
$\int_0^T z(t) \mathrm{d}W(t) := \lim_{\omega \rightarrow 0} \sum_{i=1}^{\frac{T}{\omega}} z(i\omega)\times[W((i+1)\omega) -W(i\omega)].$   
\end{definition}
\begin{lemma}[\bf It\^o's Lemma, integral form with no drift; Theorem 3.7.1 of \cite{lawler2010stochastic}] \label{lemma_ito_lemma_new}

Let $f:  \mathbb{R}^n \rightarrow \mathbb{R}$ be any twice-differentiable function.
Let $W(t) \in \mathbb{R}^n$  be a Brownian motion, and let $X(t)  \in \mathbb{R}^n$ be an It\^o diffusion process with mean zero defined by the following stochastic differential equation:
\begin{equation}
\mathrm{d}X_j(t) = \sum_{i=1}^d R_{i j }(t) \mathrm{d}W_i(t),
\end{equation}
for some It\^o diffusion $R(t) \in \mathbb{R}^{n \times n}$ adapted to the filtration generated by the Brownian motion $W(t)$.
Then for any $T\geq 0$,
\begin{align*}
  f(X(T)) -f(X(0)) &= \int_0^T \sum_{i=1}^n \sum_{\ell=1}^n \left(\frac{\partial}{\partial X_\ell} f(X(t))\right) R_{i \ell}(t) \mathrm{d}W_i(t)\\
  &\qquad +  \frac{1}{2} \int_0^T \sum_{i=1}^n \sum_{j=1}^n \sum_{\ell=1}^n \left(\frac{\partial^2}{\partial X_{j} \partial X_\ell} f(X(t))\right) R_{i j }(t) R_{i \ell}(t) \mathrm{d}t.
\end{align*}
\end{lemma}

\paragraph{Dyson Brownian motion.}
Let $W(t) \in \mathbb{R}^{d \times d}$ be a matrix where each entry is an independent standard Brownian motion with distribution $N(0, tI_d)$ at time $t$, and let $B(t) = W(t) + W^\top(t)$.
Define the symmetric-matrix valued stochastic process $\Phi(t)$ as follows:
\begin{equation} \label{eq_DBM_matrix}
    \Phi(t):= M + B(t) \qquad  \forall t\geq 0.
\end{equation}
\noindent The process $\Phi(t)$ is referred to as (matrix) Dyson Brownian motion.
At every time $t>0$ the eigenvalues $\gamma_1(t), \ldots, \gamma_d(t)$ of $\Phi(t)$ are distinct with probability $1$, and \eqref{eq_DBM_matrix} induces a stochastic process on the eigenvalues and eigenvectors.
The process on the eigenvalues and eigenvectors can be expressed via the following diffusion equations.
The eigenvalue diffusion process, which is also referred to as (eigenvalue) ``Dyson Brownian motion'', is defined by the stochastic differential equation \eqref{eq_DBM_eigenvalues}.
The  (eigenvalue) Dyson Brownian motion is an It\^o diffusion and can be expressed can be expressed by the following stochastic differential equation \cite{dyson1962brownian}:
\begin{equation} \label{eq_DBM_eigenvalues}
    \mathrm{d} \gamma_i(t) = \mathrm{d}B_{i i}(t) + \sum_{j \neq i} \frac{1}{\gamma_i(t) - \gamma_j(t)} \mathrm{d}t \qquad \qquad \forall i \in [d], t > 0.
\end{equation}
\noindent The corresponding eigenvector process $v_1(t), \ldots, v_d(t)$, referred to as the Dyson vector flow, is also an It\^o diffusion and,  conditional on the eigenvalue process \eqref{eq_DBM_eigenvalues}, can be expressed by the following stochastic differential equation (see e.g.,  \cite{anderson2010introduction}):
\begin{equation} \label{eq_DBM_eigenvectors}
    \mathrm{d}u_i(t) = \sum_{j \neq i} \frac{\mathrm{d}B_{ij}(t)}{\gamma_i(t) - \gamma_j(t)}u_j(t) - \frac{1}{2}\sum_{j \neq i} \frac{\mathrm{d}t}{(\gamma_i(t)- \gamma_j(t))^2}u_i(t) \qquad \qquad \forall i \in [d], t > 0.
\end{equation}
\paragraph{Eigenvalue bounds.}
The following two Lemmas will help us bound the gaps in the eigenvalues of the Dyson Brownian motion:

\begin{lemma} [Theorem 4.4.5 of \cite{vershynin2018high}, special case \footnote{The theorem is stated for sub-Gaussian entries in terms of a constant $C$; this constant is $C=2$ in the special case where the entries are $N(0,1)$ Gaussian.}] \label{lemma_concentration} 
Let $W \in \mathbb{R}^{d \times d}$ with i.i.d. $N(0,1)$ entries. Then 
 $   \mathbb{P}(\|W\|_2 > 2(\sqrt{d} +s) <  2e^{-s^2}$
 for any $s>0$.%
\end{lemma}

\begin{lemma}[\bf Weyl's Inequality; \cite{bhatia2013matrix}]\label{lemma_weyl}
If $A,B \in \mathbb{R}^{d\times d}$ are two symmetric matrices, and denoting the $i$'th-largest eigenvalue of any symmetric matrix $M$ by $\sigma_i(M)$, we have
%
$ \sigma_i(A) + \sigma_d(B) \leq  \sigma_i(A + B) \leq  \sigma_i(A) + \sigma_1(B).
$
\end{lemma}

\section{Proof Overview of Theorem \ref{thm_large_gap} -- Main Result}\label{sec:proof}

We give an overview of the proof of Theorem \ref{thm_large_gap}, along with the main technical lemmas used to prove this result.
Section \ref{sec_outline_of_proof} outlines the different steps in our proof.
In Steps 1 and 2 we construct the matrix-valued diffusion used in our proof.
Steps 3,4, and 5 present the main technical lemmas,
and in step 6 we explain how to complete the proof.
The statements of the lemmas and the highlights of their proofs are given in Sections \ref{sec_computing_derivative}, \ref{sec_gap_bounds}, \ref{Sec_integration}.
In Section \ref{sec_completing_the_proof} we explain how to complete the proof.
The full proofs appear in Sections \ref{sec_proof_of_lemmas} and \ref{sec_proof_thm_large_gap}.

\subsection{Outline of proof} \label{sec_outline_of_proof}

\begin{enumerate}
[leftmargin=10pt]
    \item \textbf{Step 1: Expressing the Gaussian Mechanism as a Dyson Brownian Motion.}
To obtain our utility bound, we view the Gaussian mechanism as a matrix-valued Brownian motion \eqref{eq_DBM_matrix} initialized at the input matrix $M$:
 $ \Phi(t):= M + B(t) \qquad  \forall t\geq 0.$ 
 If we run this Brownian motion for time $T = \nicefrac{\sqrt{2\log(\frac{1.25}{\delta})}}{\eps}$ we have that $\Phi(T) = (\nicefrac{\sqrt{2\log(\frac{1.25}{\delta})}}{\eps})(G +G^\top)$, recovering the output of the Gaussian mechanism.
In other words, the input to the Gaussian mechanism is $M = \Phi(0)$, and the output is $\hat{M} = \Phi(T)$.

\item \textbf{Step 2: Expressing the post-processed mechanism as a matrix diffusion $\Psi(t)$.} Our goal is to bound $\| \hat{V} \Lambda \hat{V}^\top -  V \Lambda V^\top \|_F$, where $M = V \Sigma V^\top$ and $\hat{M} = \hat{V} \hat{\Sigma} \hat{V}^\top$ are  spectral decompositions of $M$ and $\hat{M}$.
To bound the error  $\| \hat{V} \Lambda \hat{V}^\top -  V \Lambda V^\top \|_F$ we will define a stochastic process $\Psi(t)$ such that $\Psi(0) =  V \Lambda V^\top$ and $\Psi(t) = \hat{V} \Lambda \hat{V}^\top$, and then bound the Frobenius distance $\|\Psi(T) - \Psi(0)\|_F$ by integrating the (stochastic) derivative of $\Psi(t)$ over the time interval $[0,T]$.

Towards this end, at every time $t$, let $\Phi(t) = U(t) \Gamma(t) U(t)^\top$  be a spectral decomposition of the symmetric matrix $\Phi(t)$,  where $\Gamma(t)$ is a diagonal matrix with diagonal entries $\gamma_1(t) \geq \cdots \geq \gamma_d(t)$ that are the eigenvalues of $\Phi(t)$, and $U(t) = [u_1(t), \ldots, u_d(t)]$ is a $d\times d$ orthogonal matrix whose columns $u_1(t), \ldots, u_d(t)$ are an orthonormal basis of eigenvectors of $\Phi(t)$.
At every time $t$, define $\Psi(t)$ to be the symmetric matrix with eigenvalues $\Lambda$ and eigenvectors given by the columns of $U(t)$: 
 $   \Psi(t):= U(t) \Lambda U(t)^\top$  $\forall t \in [0,T].$

\item \textbf{Step 3: Computing the stochastic derivative $\mathrm{d}\Psi(t)$.} To bound the expected squared Frobenius distance $\mathbb{E}[\|\Psi(T) - \Psi(0)\|_F^2]$, we first compute the stochastic derivative $\mathrm{d}\Psi(t)$ of the matrix diffusion $\Psi(T)$ (Lemma \ref{Lemma_orbit_differntial}).

\item \textbf{Step 4: Bounding the eigenvalue gaps.} 
The equation for the derivative $\mathrm{d}\Psi(t)$ includes terms with magnitude proportional to the inverse of the eigenvalue gaps $\Delta_{ij}(t):= \gamma_i(t) - \gamma_j(t)$ for each $i,j \in[d]$, which evolve over time.
In order to bound these terms, we use Weyl's inequality (Lemma \ref{lemma_weyl}) to show that w.h.p. the gaps in the top $k+1$ eigenvalues $\Delta_{ij}(t)$ satisfy $\Delta_{ij}(t) \geq \Omega(\sigma_i - \sigma_j)$ for every time $t\in [0,T]$ (Lemma \ref{lemma_gap_concentration}), provided that the initial gaps are sufficiently large (Assumption \ref{assumption_gaps})  (See Section \ref{sec_necessity_of_assumption} for a discussion on why we need this assumption for our proof to work).

\item \textbf{Step 5: Integrating the stochastic differential equation.} Next, we express the expected squared Frobenius distance $\mathbb{E}[\|\Psi(T) - \Psi(0)\|_F^2]$ as an integral
$ \|\Psi(T) - \Psi(0)\|_F^2 ]=\mathbb{E}\left[ \left\|\int_0^T \mathrm{d}\Psi(t)\right\|_F^2\right].$
We then apply It\^o's Lemma (Lemma \ref{lemma_ito_lemma_new})  to obtain a formula for this integral.
Roughly speaking, the formula we obtain (Lemma \ref{Lemma_integral}) is
\begin{equation}\label{eq_g3} \mathbb{E}\left[\left\|\Psi(T) -  \Psi(0)\right \|_F^2 \right] \approx  \int_0^{T}   \mathbb{E}\left[ \sum_{i=1}^{d}  \sum_{j \neq i}  \frac{(\lambda_i - \lambda_j)^2}{\Delta^2_{ij}(t)} \right] \mathrm{d}t  
     +    T \int_0^{T}\mathbb{E}\left[\sum_{i=1}^{d}\left(\sum_{j\neq i} \frac{\lambda_i - \lambda_j}{\Delta^2_{ij}(t)}\right)^2   \right]\mathrm{d}t
\end{equation}

\item \textbf{Step 6: Completing the proof.}  Plugging the bound $\Delta_{ij}(t) \geq \Omega(\sigma_i - \sigma_j)$ into \eqref{eq_g3}, and noting that the first term on the r.h.s. of \eqref{eq_g3} is at least as large as the second term since $\sigma_i - \sigma_j \geq \sqrt{d}$,
we obtain the bound in Theorem \ref{thm_large_gap}.

\end{enumerate}

\subsection{Step 3: Computing the stochastic derivative $\mathrm{d}\Psi(t)$}\label{sec_computing_derivative}

$\Psi(t)$ is itself a matrix-valued diffusion. 
We use the eigenvalue and eigenvector dynamics \ref{eq_DBM_eigenvalues} and \ref{eq_DBM_eigenvectors} together with It\^o's Lemma (Lemma \ref{lemma_ito_lemma_new}) to compute the It\^o derivative of this diffusion.
Towards this end, we first decompose the matrix $\Psi(t)$ as a sum of its eigenvectors:  $\Psi(t) =   \sum_{i=1}^{d}  \lambda_i u_i(t) u_i^\top(t)$.
Thus, we have   
\begin{equation} \label{eq_g4}
 \mathrm{d}\Psi(t) =   \sum_{i=1}^{d}  \lambda_i \mathrm{d}(u_i(t) u_i^\top(t)).
\end{equation}
We begin by computing the stochastic derivative $\mathrm{d}(u_i(t) u_i^\top(t))$ for each $i \in [d]$, by applying the formula for the derivative of $u_i(t)$ in  \eqref{eq_DBM_eigenvectors}, together with It\^o's Lemma (Lemma \ref{lemma_ito_lemma_new}):

\begin{lemma}[\bf Stochastic derivative of $u_i(t) u_j^\top(t)$]\label{Lemma_projection_differntial}
For all $t \in [0,T]$, 
\begin{align*}
\mathrm{d}(u_i(t) u_i^\top(t))=        \sum_{j \neq i} &\frac{\mathrm{d}B_{ij}(t)}{\gamma_i(t) - \gamma_j(t)}(u_i(t) u_j^\top(t) + u_j(t) u_i^\top(t))\\
& + \sum_{j \neq i} \frac{\mathrm{d}t}{(\gamma_i(t)- \gamma_j(t))^2} (u_i(t) u_i^\top(t) - u_j(t)u_j^\top(t)).
\end{align*}
\end{lemma}
\noindent
The proof is in Section \ref{sec_proof_of_lemmas}.
 Plugging Lemma \ref{Lemma_projection_differntial} into \eqref{eq_g4},   we get an expression for $\mathrm{d}\Psi(t)$:

\begin{lemma}[\bf Stochastic derivative of $\Psi(t)$; see Section \ref{sec_proof_of_lemmas} for proof]\label{Lemma_orbit_differntial}
For all $t \in [0,T]$ we have that
$  \mathrm{d}\Psi(t)    =   \frac{1}{2}\sum_{i=1}^{d} \sum_{j \neq i} (\lambda_i - \lambda_j) \frac{\mathrm{d}B_{ij}(t)}{\gamma_i(t) - \gamma_j(t)}(u_i(t) u_j^\top(t) + u_j(t) u_i^\top(t))
 \textstyle     + \sum_{i=1}^{d} \sum_{j\neq i} (\lambda_i - \lambda_j) \frac{\mathrm{d}t}{(\gamma_i(t)- \gamma_j(t))^2} u_i(t) u_i^\top(t).$

\end{lemma}

\subsection{Step 4: Bounding the eigenvalue gaps} \label{sec_gap_bounds}

The derivative in Lemma \ref{Lemma_orbit_differntial} contains terms with magnitude proportional to the inverse of the eigenvalue gaps $\Delta_{ij}(t):= \gamma_i(t) - \gamma_j(t)$.
To bound these terms, we would like to show that $\inf_{t\in [0,T]} \Delta_{ij}(t) \geq \Omega(\sigma_i - \sigma_j)$ for each $i<j \leq  k+1$ with high probability.
Towards this end, we first apply the spectral norm concentration  bound for Gaussian random matrices (Lemma \ref{lemma_concentration}), which provides a high-probability bound for $\|B(t)\|_2$ at any time $t$, together with Doob's submartingale inequality, to show that the spectral norm of the matrix-valued Brownian motion $B(t)$ does not exceed $T\sqrt{d}$ at any time $t\in [0,T]$ w.h.p.:
\begin{lemma}[\bf Spectral norm bound] \label{lemma_spectral_martingale}
 For every $T>0$, we have,
      $$ \mathbb{P}\left(\sup_{t \in [0,T]}\|B(t)\|_2 > 2T\sqrt{d} + \alpha)\right) \leq 2\sqrt{\pi} e^{-\frac{1}{8}\frac{\alpha^2}{T^2}}.$$

\end{lemma}
\noindent
The proof appears in Section \ref{sec_proof_of_lemmas}.
Next, we use Lemma \ref{lemma_spectral_martingale} to bound the eigenvalue gaps:

\begin{lemma}[\bf Eigenvalue gap bound]\label{lemma_gap_concentration}
Whenever  $\gamma_i(0) - \gamma_{i+1}(0) \geq 4 T \sqrt{d}$ for every $i \in S$ and $T>0$ and some subset $S \subset [d-1]$, we have $$\mathbb{P}\left( \bigcup_{i\in S} \left\{\inf_{t \in [0,T]} \gamma_i(t) - \gamma_{i+1}(t) <  \frac{1}{2}(\gamma_i(0) - \gamma_{i+1}(0)) - \alpha) \right\}\right) \leq  2\sqrt{\pi} e^{-\frac{1}{32} \alpha^2}.$$
\end{lemma}
\noindent
To prove Lemma \ref{lemma_gap_concentration}, we plug Lemma \ref{lemma_spectral_martingale} into Weyl's Inequality (Lemma \ref{lemma_weyl}),  to show that 
\begin{equation*}
    \gamma_i(t) - \gamma_{i+1}(t) \geq \sigma_i - \sigma_{i+1} - \|B(t)\|_2 \geq \Omega(\sigma_i - \sigma_{i+1} - T \sqrt{d}) \geq \Omega(\sigma_i - \sigma_{i+1}),
\end{equation*}
with high probability for each $i\leq  k$ (Lemma \ref{lemma_gap_concentration}).
The last inequality holds since Assumption \ref{assumption_gaps} ensures $\sigma_i - \sigma_{i+1} \geq \frac{1}{2}T \sqrt{d}$ for  $i\leq  k$.
The full proof is in Section \ref{sec_proof_of_lemmas}.

\subsection{Step 5: Integrating the stochastic differential equation} \label{Sec_integration}

Next, we would like to integrate the derivative $\mathrm{d}\Psi(t)$ to obtain an expression for $\mathbb{E}[\|\Psi(T) - \Psi(0)\|_F^2]$, and to then plug in our high-probability bounds (Lemma \ref{lemma_gap_concentration}) for the gaps $\Delta_{ij}(t)$.
To allow us to later plug in these high-probability bounds after we integrate and take the expectation, we define a new diffusion process $Z_\eta(t)$ which has nearly the same stochastic differential equation as \ref{Lemma_orbit_differntial}, except that each eigenvalue gap $\Delta_{ij}(t)$ is not permitted to become smaller than the value $\eta_{ij} = \frac{1}{4}(\sigma_i - \max(\sigma_{j}, \sigma_{k+1}))$ for each $i<j$.
Towards this end, fix any  $\eta \in \mathbb{R}^{d\times d}$,  define the following matrix-valued It\^o diffusion $Z_\eta(t)$ via its It\^o derivative  $\mathrm{d}Z_\eta(t)$:
\begin{eqnarray} \label{eq_orbit_diffusion2}
      \textstyle    \mathrm{d}Z_\eta(t)   & :=   \frac{1}{2}\sum_{i=1}^{d} \sum_{j \neq i} |\lambda_i - \lambda_j| \frac{\mathrm{d}B_{ij}(t)}{\max(|\Delta_{ij}(t)|, \eta_{ij})}(u_i(t) u_j^\top(t) + u_j(t) u_i^\top(t)) \nonumber \\
     &  + \sum_{i=1}^{d} \sum_{j\neq i} (\lambda_i - \lambda_j) \frac{\mathrm{d}t}{\max(\Delta^2_{ij}(t), \eta_{ij}^2)} u_i(t) u_i^\top(t),
\end{eqnarray}
with  initial condition $Z_\eta(0):= \Psi(0)$.
Thus, $Z_\eta(t) = \Psi(0)+ \int_0^t \mathrm{d}Z_\eta(s)$ for all $t \geq 0$.
We then integrate  $\mathrm{d}Z_\eta(t)$ over the time interval $[0,T]$, and apply  It\^o's Lemma (Lemma \ref{lemma_ito_lemma_new}) to obtain an expression for the Frobenius norm of this integral:

\begin{lemma}[\bf Frobenius distance integral] \label{Lemma_integral}

For any $T>0$, 
$\mathbb{E}\left[\left\|Z_\eta\left(T\right) -  \textstyle Z_\eta(0)\right \|_F^2 \right]=$
 \begin{eqnarray*} \textstyle 
 2\int_0^{T}   \mathbb{E}\left[ \sum_{i=1}^{d}  \sum_{j \neq i}  \frac{(\lambda_i - \lambda_j)^2}{\max(\Delta^2_{ij}(t), \eta_{ij}^2)} \mathrm{d}t \right]
      +    T \int_0^{T}\mathbb{E}\left[\sum_{i=1}^{d}\left(\sum_{j\neq i} \frac{\lambda_i - \lambda_j}{\max(\Delta^2_{ij}(t), \eta_{ij}^2)}\right)^2   \right]\mathrm{d}t.
\end{eqnarray*}
\end{lemma}
\noindent
To prove Lemma \ref{Lemma_integral}, we write
\begin{align} \label{eq_z2}
Z_\eta\left(T\right) -  Z_\eta(0)
&=  \frac{1}{2}  \int_0^{T}\sum_{i=1}^{d} \sum_{j \neq i} |\lambda_i - \lambda_j| \frac{\mathrm{d}B_{ij}(t)}{\max(|\Delta_{ij}(t)|, \eta_{ij})}(u_i(t) u_j^\top(t) + u_j(t) u_i^\top(t)) \nonumber\\
    & -   \int_0^{T}\sum_{i=1}^{d} \sum_{j\neq i} (\lambda_i - \lambda_j) \frac{\mathrm{d}t}{\max(\Delta^2_{ij}(t), \eta_{ij}^2)} u_i(t) u_i^\top(t).
\end{align}
To compute the Frobenius norm of the first term on the r.h.s. of \eqref{eq_z2}, we use It\^o's Lemma (Lemma \ref{lemma_ito_lemma_new}), with  $X(t):=  \int_0^{t}\sum_{i=1}^{d} \sum_{j \neq i} |\lambda_i - \lambda_j| \frac{\mathrm{d}B_{ij}(s)}{\max(|\Delta_{ij}(s)|, \eta_{ij})}(u_i(s) u_j^\top(s) + u_j(s) u_i^\top(s))$ and the function $f(X):= \|X\|_F^2 = \sum_{i=1}^d \sum_{j=1}^d X_{ij}^2$.
By It\^o's Lemma, we have
 \begin{align} \label{eq_z1}
   &\mathbb{E}[\|X(T)\|_F^2 - \|X(0)\|_F^2] 
    = \mathbb{E} [\frac{1}{2} \int_0^t \sum_{\ell, r} \sum_{\alpha, \beta} (\frac{\partial}{ \partial X_{\alpha \beta}} f(X(t))) R_{(\ell r) (\alpha \beta)}(t) \mathrm{d}B_{\ell r}(t) ] \nonumber \\
     &\qquad \qquad +\mathbb{E}\left [\frac{1}{2} \int_0^t \sum_{\ell, r} \sum_{i,j} \sum_{\alpha, \beta} \left(\frac{\partial^2}{\partial X_{ij} \partial X_{\alpha \beta}} f(X(t))\right) R_{(\ell r) (i j)}(t)  R_{(\ell r) (\alpha \beta)}(t) \mathrm{d}t \right],
     \end{align}
where $R_{(\ell r) (i j)}(t) := \left(\frac{ |\lambda_i - \lambda_j| }{\max(|\Delta_{ij}(t)|, \eta_{ij})}(u_i(t) u_j^\top(t) + u_j(t) u_i^\top(t)) \right)[\ell, r]$, and where we denote by either $H_{\ell r}$ or $H[\ell, r]$  the $(\ell, r)$'th entry of any matrix $H$.

The first term on the r.h.s. of \eqref{eq_z1} is equal to zero since $\mathrm{d}B_{\ell r}(s)$ is independent of both $X(t)$ and $R(t)$ for all $s \geq t$ and the time-integral of each Brownian motion increment $\mathrm{d}B_{\alpha \beta}(s)$ has zero mean.
 To compute the second term on the r.h.s. of \eqref{eq_z1}, we use the fact that $\frac{\partial^2 }{\partial X_{ij} \partial X_{\alpha \beta}}f(X)$ is equal to 2 for $i = j$ and $0$ for $i \neq j$

To compute the Frobenius norm of the second term on the r.h.s. of \eqref{eq_z2}, we use the Cauchy-Schwarz inequality. 
The full proof appears in Section \ref{sec_proof_of_lemmas}.

\subsection{Step 6: Completing the proof} \label{sec_completing_the_proof}
To complete the proof, we plug in the high-probability bounds on the eigenvalue gaps from Section \ref{sec_gap_bounds} into Lemma \ref{Lemma_integral}.
Since by Lemma \ref{lemma_gap_concentration} $\Delta_{ij}(t) \geq \frac{1}{2}(\sigma_i - \sigma_j)$ w.h.p. for each $i,j \leq k+1$, and $\eta_{ij} = \frac{1}{4}(\sigma_i - \max(\sigma_{j}, \sigma_{k+1}))$, we must also have that $Z_\eta(t) = \Psi(t)$ for all $t \in [0, T]$ w.h.p.
Plugging in the high-probability bounds $\Delta_{ij}(t) \geq \frac{1}{2}(\sigma_i - \sigma_j)$ for each $i,j \geq k+1$, and noting that $\lambda_i - \lambda_j = 0$ for all $i,j >k$, we get that
  \begin{align}\label{eq_g5} &\mathbb{E}\left[\| \hat{V} \Lambda \hat{V}^\top -  V \Lambda V^\top \|_F^2\right] = \mathbb{E}\left[\left\|\Psi\left(T\right) -  \textstyle \Psi(0)\right \|_F^2 \right] \nonumber\\
 &\leq   2\int_0^{T} \mathbb{E}\left[ \sum_{i=1}^{d}  \sum_{j \neq i}  \frac{(\lambda_i - \lambda_j)^2}{(\sigma_i - \sigma_j)^2}\right]  \mathrm{d}t 
      +    T \int_0^{T}\mathbb{E}\left[\sum_{i=1}^{d}\left(\sum_{j\neq i} \frac{\lambda_i - \lambda_j}{(\sigma_i - \sigma_j)^2}\right)^2   \right]\mathrm{d}t \nonumber\\
      &\leq T\sum_{i=1}^{k}  \sum_{j = i+1}^d  \frac{(\lambda_i - \lambda_j)^2}{(\sigma_i - \max(\sigma_j, \sigma_{k+1}))^2}  + T^2\sum_{i=1}^{k}\left(\sum_{j = i+1}^d  \frac{\lambda_i - \lambda_j}{(\sigma_i - \max(\sigma_j, \sigma_{k+1}))^2}\right)^2.
\end{align}
Since $(\sigma_i - \max(\sigma_j, \sigma_{k+1})
\geq \Omega(\sqrt{d})$ for all $i\leq k$ and $j\in [d]$, we  can use the Cauchy-Schwarz inequality to show that the second term is (up to a factor of T) smaller than the first term: $\sum_{i=1}^{k}\left(\sum_{j = i+1}^d  \frac{\lambda_i - \lambda_j}{(\sigma_i - \max(\sigma_j. \sigma_{k+1}))^2}\right)^2  \leq \sum_{i=1}^{k} \sum_{j=i+1}^d \frac{(\lambda_i - \lambda_j)^2}{(\sigma_i-\max(\sigma_j, \sigma_{k+1})}$.  Plugging $T = \frac{\sqrt{2\log(\frac{1.25}{\delta})}}{\eps}$  into \eqref{eq_g5}, we obtain the bound in Theorem \ref{thm_large_gap}.
For the full proof of Theorem \ref{thm_large_gap}, see Section \ref{sec_proof_thm_large_gap}

\section{Proof of Lemmas}  \label{sec_proof_of_lemmas} 

\begin{proof}[Proof of Lemma \ref{Lemma_projection_differntial}]
We compute the stochastic derivative $\mathrm{d}(u_i(t) u_i^\top(t))$ by applying the formula \eqref{eq_DBM_eigenvectors}  for the stochastic derivative $\mathrm{d}u_i(t)$ of the eigenvector $u_i(t)$ in Dyson Brownian motion.
For any $t \in [0,T]$, we have that the stochastic derivative $\mathrm{d}(u_i(t) u_i^\top(t))$ satisfies
\begin{align}\label{eq_eq_derivative1}
  &\mathrm{d}(u_i(t) u_i^\top(t))
  =(u_i(t) + \mathrm{d}u_i(t))(u_i(t) + \mathrm{d}u_i(t))^\top - u_i(t)u_i(t)^\top \nonumber\\
  &= \left(u_i(t)+ \sum_{j \neq i} \frac{\mathrm{d}B_{ij}(t)}{\gamma_i(t) - \gamma_j(t)}u_j(t) - \frac{1}{2}\sum_{j \neq i} \frac{\mathrm{d}t}{(\gamma_i(t)- \gamma_j(t))^2}u_i(t) \right) \nonumber\\
  &\qquad\times \left(u_i(t) + \sum_{j \neq i} \frac{\mathrm{d}B_{ij}(t)}{\gamma_i(t) - \gamma_j(t)}u_j(t) - \frac{1}{2}\sum_{j \neq i} \frac{\mathrm{d}t}{(\gamma_i(t)- \gamma_j(t))^2}u_i(t) \right)^\top - u_i(t)u_i(t)^\top \nonumber\\
  &= u_i(t) u_i^\top(t) + \sum_{j \neq i} \frac{\mathrm{d}B_{ij}(t)}{\gamma_i(t) - \gamma_j(t)}(u_i(t) u_j^\top(t) + u_j(t) u_i^\top(t))  - \sum_{j\neq i} \frac{\mathrm{d}t}{(\gamma_i(t)- \gamma_j(t))^2} u_i(t) u_i^\top(t)\nonumber\\
  & \qquad + \sum_{j \neq i}  \sum_{\ell \neq i} \frac{\mathrm{d}B_{ij}(t)\mathrm{d}B_{i\ell}(t)}{(\gamma_i(t) - \gamma_j(t)) (\gamma_i(t) - \gamma_\ell(t))} u_j(t)u_{\ell}^\top(t) \nonumber\\
  & \qquad   \qquad  \qquad -\varphi_1(t)\varphi_2(t)^\top  -\varphi_2(t)\varphi_1(t)^\top + -\varphi_2(t)\varphi_2(t)^\top - u_i(t)u_i(t)^\top,
  \end{align}
  where we define $\varphi_1(t):= \sum_{j \neq i} \frac{\mathrm{d}B_{ij}(t)}{\gamma_i(t) - \gamma_j(t)}u_j(t)$ and $\varphi_2(t):=\sum_{j \neq i} \frac{\mathrm{d}t}{(\gamma_i(t)- \gamma_j(t))^2}u_i(t)$.
The terms $\varphi_1(t) \varphi_2(t)^\top$ and $\varphi_2(t) \varphi_1(t)^\top$ have differentials $O(\mathrm{d}B_{ij} \mathrm{d}t)$, and $\varphi_2(t) \varphi_2(t)^\top$ has differentials $O(\mathrm{d}t^2)$; thus, all three terms vanish in the stochastic derivative by Lemma \ref{lemma_ito_lemma_new}.
Therefore, \eqref{eq_eq_derivative1} implies that the stochastic derivative $\mathrm{d}(u_i(t) u_i^\top(t))$ satisfies
  \begin{align}\label{eq_eq_derivative2}
  &\mathrm{d}(u_i(t) u_i^\top(t)) = \sum_{j \neq i} \frac{\mathrm{d}B_{ij}(t)}{\gamma_i(t) - \gamma_j(t)}(u_i(t) u_j^\top(t) + u_j(t) u_i^\top(t)) - \sum_{j\neq i} \frac{\mathrm{d}t}{(\gamma_i(t)- \gamma_j(t))^2} u_i(t) u_i^\top(t)\nonumber\\
  & \qquad + \sum_{j \neq i}  \sum_{\ell \neq i} \frac{\mathrm{d}B_{ij}(t)\mathrm{d}B_{i\ell}(t)}{(\gamma_i(t) - \gamma_j(t)) (\gamma_i(t) - \gamma_\ell(t))} u_j(t)u_{\ell}^\top(t)\nonumber\\
    &= \sum_{j \neq i} \frac{\mathrm{d}B_{ij}(t)}{\gamma_i(t) - \gamma_j(t)}(u_i(t) u_j^\top(t) + u_j(t) u_i^\top(t)) - \sum_{j\neq i} \frac{\mathrm{d}t}{(\gamma_i(t)- \gamma_j(t))^2} u_i(t) u_i^\top(t)\nonumber\\
  & \qquad + \sum_{j \neq i} \left(\frac{\mathrm{d}B_{ij}(t)}{\gamma_i(t) - \gamma_j(t)} \right)^2 u_j(t)u_j^\top(t)\nonumber\\
      &=  \sum_{j \neq i} \frac{\mathrm{d}B_{ij}(t)}{\gamma_i(t) - \gamma_j(t)}(u_i(t) u_j^\top(t) + u_j(t) u_i^\top(t)) - \sum_{j\neq i} \frac{\mathrm{d}t}{(\gamma_i(t)- \gamma_j(t))^2} u_i(t) u_i^\top(t)\nonumber\\
  & \qquad + \sum_{j \neq i} \frac{\mathrm{d}t}{(\gamma_i(t) - \gamma_j(t))^2} u_j(t)u_j^\top(t),
\end{align}
where the second-to-last equality holds since all terms $\mathrm{d}B_{ij}(t)\mathrm{d}B_{i\ell}(t)$ with $j \neq \ell$ in the sum $\sum_{j \neq i}  \sum_{\ell \neq i} \frac{\mathrm{d}B_{ij}(t)\mathrm{d}B_{i\ell}(t)}{(\gamma_i(t) - \gamma_j(t)) (\gamma_i(t) - \gamma_\ell(t))} u_j(t)u_{\ell}^\top(t)$  vanish by It\^o's Lemma (Lemma \ref{lemma_ito_lemma_new}) since they have mean 0 and are $O(\mathrm{d}B_{ij}(t) \mathrm{d}B_{i\ell}(t))$; we are therefore left only with the terms $j = \ell$ in the sum which have differential terms $(\mathrm{d}B_{ij}(t))^2$ which have mean $\mathrm{d}t$ plus higher-order terms which vanish by It\^o's Lemma. 
Therefore  \eqref{eq_eq_derivative2} implies that
\begin{align*}
    &\mathrm{d}(u_i(t) u_i^\top(t))\\ 
    & = \sum_{j \neq i} \frac{\mathrm{d}B_{ij}(t)}{\gamma_i(t) - \gamma_j(t)}(u_i(t) u_j^\top(t) + u_j(t) u_i^\top(t)) - \sum_{j\neq i} \frac{\mathrm{d}t}{(\gamma_i(t)- \gamma_j(t))^2} (u_i(t) u_i^\top(t) - u_j(t)u_j^\top(t)).
\end{align*}
\end{proof}


\begin{proof}[Proof of Lemma \ref{Lemma_orbit_differntial}]
To compute the stochastic derivative of $\Psi(t)$, we would like to apply our formula for the stochastic derivative of the projection matrix $u_i(t) u_i^\top(t)$ for each eigenvector $u_i(t)$ (Lemma \ref{Lemma_projection_differntial}).
Towards this end, we first decompose the matrix $\Psi(t)$ as a sum of these projection matrices $u_i(t) u_i^\top(t)$:
\begin{equation}\label{eq_derivative3}
\textstyle \Psi(t) =   \sum_{i=1}^{d}  \lambda_i (u_i(t) u_i^\top(t)).
\end{equation}
Taking the derivative on both sides of \eqref{eq_derivative3}, we have
\begin{equation} \label{eq_derivative4}
\textstyle \mathrm{d}\Psi(t) =   \sum_{i=1}^{d}  \lambda_i \mathrm{d}(u_i(t) u_i^\top(t)).
\end{equation}
Thus, plugging in Lemma \ref{Lemma_projection_differntial} for each $i \in[d]$ into \eqref{eq_derivative4}, we have that
\begin{align*}
  \mathrm{d}\Psi(t) &=   \sum_{i=1}^{d}  \lambda_i \mathrm{d}(u_i(t) u_i^\top(t))\\
   &\stackrel{\textrm{Lemma \ref{Lemma_projection_differntial}}}{=} \sum_{i=1}^{d}  \lambda_i  \left( \sum_{j \neq i} \frac{\mathrm{d}B_{ij}(t)}{\gamma_i(t) - \gamma_j(t)}(u_i(t) u_j^\top(t) + u_j(t) u_i^\top(t))\right. \\
   &- \left.\sum_{j\neq i} \frac{\mathrm{d}t}{(\gamma_i(t)- \gamma_j(t))^2} (u_i(t) u_i^\top(t) - u_j(t)u_j^\top(t)) \right)\\
   & =  \frac{1}{2}\sum_{i=1}^{d} \sum_{j \neq i} (\lambda_i - \lambda_j) \frac{\mathrm{d}B_{ij}(t)}{\gamma_i(t) - \gamma_j(t)}(u_i(t) u_j^\top(t) + u_j(t) u_i^\top(t))\\
   & - \frac{1}{2} \sum_{i=1}^{d} \sum_{j\neq i} (\lambda_i - \lambda_j) \frac{\mathrm{d}t}{(\gamma_i(t)- \gamma_j(t))^2} (u_i(t) u_i^\top(t) - u_j(t)u_j^\top(t))\\
      & =  \frac{1}{2}\sum_{i=1}^{d} \sum_{j \neq i} (\lambda_i - \lambda_j) \frac{\mathrm{d}B_{ij}(t)}{\gamma_i(t) - \gamma_j(t)}(u_i(t) u_j^\top(t) + u_j(t) u_i^\top(t))\\
    & - \sum_{i=1}^{d} \sum_{j\neq i} (\lambda_i - \lambda_j) \frac{\mathrm{d}t}{(\gamma_i(t)- \gamma_j(t))^2} u_i(t) u_i^\top(t),
\end{align*}
where the second equality holds by Lemma \ref{Lemma_projection_differntial}.
To see why the third equality holds, for the first term inside the summation, note that  $\mathrm{d}B(t)$ is a symmetric matrix of differentials which means that $\mathrm{d}B_{ij}(t) = \mathrm{d}B_{ji}(t)$ for all $i,j$, and hence that  $\frac{\mathrm{d}B_{ij}(t)}{\gamma_i(t) - \gamma_j(t)}(u_i(t) u_j^\top(t) + u_j(t) u_i^\top(t)) = - \frac{\mathrm{d}B_{ji}(t)}{\gamma_j(t) - \gamma_i(t)}(u_j(t) u_i^\top(t) + u_i(t) u_j^\top(t))$ for all $i \neq j$.
For the second term inside the summation, note that $\frac{\mathrm{d}t}{(\gamma_i(t)- \gamma_j(t))^2} (u_i(t) u_i^\top(t) - u_j(t)u_j^\top(t)) = - \frac{\mathrm{d}t}{(\gamma_j(t)- \gamma_i(t))^2} (u_j(t) u_j^\top(t) - u_i(t)u_i^\top(t))$ for all $i \neq j$.
\end{proof}


\begin{proof}[Proof of Lemma \ref{lemma_spectral_martingale}]
To prove Lemma \ref{lemma_spectral_martingale} we will use Doob's submartingale inequality.
Towards this end, let  $\mathcal{F}_s$ be the filtration generated by $B(s)$.
First, we note that $\exp(\|B(t)\|_2)$ is a submartingale for all $t \geq 0$; that is, $\mathbb{E}[\exp(\|B(t)\|_2) |  \mathcal{F}_s] \geq  \mathrm{exp}(\| B(s)\|_2 )$ for all $0\leq s \leq t$.
This is because for all $s \leq t$, we have
\begin{align*}
\mathbb{E}[\exp(\|B(t)\|_2) |  \mathcal{F}_s] &= \mathbb{E}\left[\mathrm{exp}\left(\sup_{v \in \mathbb{R}^d: \|v\|_2 = 1} v^\top B(t) v\right) \, \, \bigg | \, \,\mathcal{F}_s\right]\\
&\geq \mathrm{exp}\left(\mathbb{E}\left[\sup_{v \in \mathbb{R}^d: \|v\|_2 = 1} v^\top B(t) v \, \, \bigg | \, \,\mathcal{F}_s\right]\right)\\
&\geq \mathrm{exp}\left(\sup_{v \in \mathbb{R}^d: \|v\|_2 = 1}  \mathbb{E}\left[v^\top B(t) v \, \,  | \, \, \mathcal{F}_s\right]\right)\\
&=  \mathrm{exp}\left(\sup_{v \in \mathbb{R}^d: \|v\|_2 = 1}  \mathbb{E}\left[v^\top (B(t) - B(s)) v  + v^\top B(s) v \, \,  | \, \, \mathcal{F}_s\right] \right)\\
&=  \mathrm{exp}\left(\sup_{v \in \mathbb{R}^d: \|v\|_2 = 1}  \mathbb{E}\left[v^\top (B(t) - B(s)) v | \, \, \mathcal{F}_s\right]  +  \mathbb{E}\left[v^\top B(s) v \, \,  | \, \, \mathcal{F}_s\right] \right)\\
&= \mathrm{exp}\left( \sup_{v \in \mathbb{R}^d: \|v\|_2 = 1}  \mathbb{E}\left[v^\top B(s) v \, \,  | \, \, \mathcal{F}_s\right] \right)\\
&= \mathrm{exp}\left(\sup_{v \in \mathbb{R}^d: \|v\|_2 = 1}  v^\top B(s) v \right),\\
&= \mathrm{exp}(\| B(s)\|_2 ), 
\end{align*}
where the first inequality holds by Jensen's inequality since $\exp(\cdot)$ is convex, and the third equality holds since $v^\top (B(t) - B(s)) v$ is independent of $\mathcal{F}_s$ and is distributed as $N(0,2(t-s))$.
Thus, by Doob's submartingale inequality, for any $\beta>0$  (we will choose the value of $\beta$ later to optimize our bound) we have,
\begin{align*}
    \mathbb{P}\left(\sup_{t \in [0,T]}\|B(t)\|_2 > 2T(\sqrt{d} + \alpha)\right) & =      \mathbb{P}\left(\sup_{t \in [0,T]}\frac{\beta}{2T}\|B(t)\|_2 - \beta\sqrt{d} > \beta \alpha\right) \\
    &=\mathbb{P}\left(\sup_{t \in [0,T]}\exp\left(\frac{\beta}{2T}\|B(t)\|_2 - \beta\sqrt{d}\right) > \exp(\beta \alpha)\right)\\
    & \leq \frac{\mathbb{E}[\exp(\frac{\beta}{2T}\|B(t)\|_2 - \beta\sqrt{d})]}{\exp(\beta\alpha)}\\
        & = \frac{\int_{0}^{\infty} \mathbb{P}[\exp(\frac{\beta}{2T}\|B(t)\|_2 - \beta\sqrt{d})>x] \mathrm{d}x}{\exp(\beta \alpha)}\\
                & = \frac{\int_{0}^{\infty} \mathbb{P}[\frac{1}{2}\|B(t)\|_2 - \sqrt{d}> \beta^{-1} \log(x)] \mathrm{d}x}{\exp(\beta \alpha)}\\
                & \leq \frac{\int_{0}^{\infty} 2 e^{-\beta^{-2}\log^2(x)} \mathrm{d}x}{\exp(\beta \alpha)}\\
                                                                                               & =  \frac{2\sqrt{\pi} \beta e^{\frac{1}{4}\beta^2}}{\exp(\beta \alpha)}\\
                                                                & \leq  \frac{2\sqrt{\pi}e^{\frac{1}{2}\beta^2}}{\exp(\beta \alpha)}\\
                &= 2\sqrt{\pi} e^{\frac{1}{2}\beta^2 - \beta \alpha},
\end{align*}
where the first inequality holds by Doob's submartingale inequality, and the second inequality holds by Lemma \ref{lemma_concentration}.
Setting $\beta = \alpha$, we have
\begin{eqnarray*}
    \textstyle \mathbb{P}\left(\sup_{t \in [0,T]}\|B(t)\|_2 > T(\sqrt{d} + \alpha)\right) \leq 2\sqrt{\pi} e^{-\frac{1}{2}\alpha^2}.
    \end{eqnarray*}
\end{proof}

%
%

\begin{proof}[Proof of Lemma \ref{lemma_gap_concentration}]
To prove Lemma \ref{lemma_gap_concentration}, we plug our high-probability concentration bound for $\sup_{t\in[0,T]} \|B(t)\|_2$ (Lemma \ref{lemma_spectral_martingale}) into Weyl's Inequality (Lemma \ref{lemma_weyl}).
Since, at every time $t$, $\Phi(t)= M + B(t)$ and  $\gamma_1(t) \geq \cdots \geq \gamma_d(t)$ are the eigenvalues of $\Phi(t)$, Weyl's Inequality implies that
\begin{equation} \label{eq_gap1}
    \gamma_i(t) - \gamma_{i+1}(t) \geq \gamma_i(0) - \gamma_{i+1}(0) - \|B(t)\|_2, \qquad \forall t\in[0,T], i \in [d].
\end{equation}
Therefore, plugging Lemma \ref{lemma_spectral_martingale} into \eqref{eq_gap1} we have that
    \begin{align*}
             &\mathbb{P}\left( \bigcup_{i\in S} \left\{\inf_{t \in [0,T]} \gamma_i(t) - \gamma_{i+1}(t) <  \frac{1}{2}(\gamma_i(0) - \gamma_{i+1}(0)) - \alpha) \right\}\right)\\
              &\stackrel{\textrm{Eq. \eqref{eq_gap1}}}{\leq}\mathbb{P}\left( \bigcup_{i\in S} \left\{ \gamma_i(0) - \gamma_{i+1}(0) - \sup_{t \in [0,T]} 2\|B(t)\|_2  <  \frac{1}{2}(\gamma_i(0) - \gamma_{i+1}(0)) - \alpha) \right\}\right)  \\
         &= \mathbb{P}\left(  \bigcup_{i\in S} \left\{ \sup_{t \in [0,T]} \|B(t)\|_2  >  \frac{1}{4}(\gamma_i(0) - \gamma_{i+1}(0)) + \frac{1}{2} \alpha) \right\}\right)  \\
                  &\stackrel{\textrm{Assumption} \ref{assumption_gaps}}{\leq} \mathbb{P}\left( \bigcup_{i\in S} \left\{ \sup_{t \in [0,T]} \|B(t)\|_2  >   2T \sqrt{d} + \frac{1}{2} \alpha) \right\}\right)\\
                  &= \mathbb{P}\left(\sup_{t \in [0,T]} \|B(t)\|_2  >   2T \sqrt{d} + \frac{1}{2} \alpha) \right)\\  
                  &\stackrel{\textrm{Lemma \ref{lemma_spectral_martingale}}}{\leq} 2\sqrt{\pi} e^{-\frac{1}{32} \alpha^2},
    \end{align*}
The first inequality holds by \eqref{eq_gap1},
and the second inequality holds by Assumption \ref{assumption_gaps} since $\gamma_i(0) = \sigma_i$ for each $i\in[d]$ because $\Phi(0) = M$.
The last inequality holds by the high-probability concentration bound for $\sup_{t\in[0,T]} \|B(t)\|_2$ (Lemma \ref{lemma_spectral_martingale}).
\end{proof}


%
%
\begin{proof}[Proof of Lemma \ref{Lemma_integral}]

%
By the definition of $Z_\eta(t)$ we have that
\begin{align*}
Z_\eta\left(T\right) -  Z_\eta(0) &= \int_0^{T} \mathrm{d}Z_\eta(t)\\
&=  \frac{1}{2}  \int_0^{T}\sum_{i=1}^{d} \sum_{j \neq i} |\lambda_i - \lambda_j| \frac{\mathrm{d}B_{ij}(t)}{\max(|\Delta_{ij}(t)|, \eta_{ij})}(u_i(t) u_j^\top(t) + u_j(t) u_i^\top(t)) \\
    & -   \int_0^{T}\sum_{i=1}^{d} \sum_{j\neq i} (\lambda_i - \lambda_j) \frac{\mathrm{d}t}{\max(\Delta^2_{ij}(t), \eta_{ij}^2)} u_i(t) u_i^\top(t).
\end{align*}
Therefore, we have that
\begin{align} \label{eq_int_1}
\left\|Z_\eta(T) -  Z_\eta(0)\right \|_F^2 &\leq  \frac{1}{2} \left\| \int_0^{T}\sum_{i=1}^{d} \sum_{j \neq i} |\lambda_i - \lambda_j| \frac{\mathrm{d}B_{ij}(t)}{\max(|\Delta_{ij}(t)|, \eta_{ij})}(u_i(t) u_j^\top(t) + u_j(t) u_i^\top(t))\right\|_F^2 \nonumber\\
    & +  \left\| \int_0^{T}\sum_{i=1}^{d} \sum_{j\neq i} (\lambda_i - \lambda_j) \frac{\mathrm{d}t}{\max(\Delta^2_{ij}(t), \eta_{ij}^2)} u_i(t) u_i^\top(t) \right\|_F^2. 
\end{align}
The first term on the r.h.s. of \eqref{eq_int_1} (inside its Frobenius norm) is a ``diffusion'' term--that is, the integral has mean 0 and Brownian motion differentials $\mathrm{d}B_{ij}(t)$ inside the integral.
The second term on the r.h.s. (inside its Frobenius norm) is a ``drift'' term-- that is, the integral has non-zero  mean and deterministic differentials $\mathrm{d}t$ inside the integral.
We bound the diffusion and drift terms separately.

\paragraph{Bounding the diffusion term:}

We first use It\^o's Lemma (Lemma \ref{lemma_ito_lemma_new}) to bound the diffusion term in \eqref{eq_int_1}.
Towards this end, let $f: \mathbb{R}^{d \times d}: \rightarrow \mathbb{R}$ be the function which takes as input a $d \times d$ matrix and outputs the square of its Frobenius norm: $f(X):= \|X\|_F^2 = \sum_{i=1}^d \sum_{j=1}^d X_{ij}^2$ for every $X \in \mathbb{R}^{d\times d}$.
Then
\begin{equation}\label{eq_int_5}
\frac{\partial^2 }{\partial X_{ij} \partial X_{\alpha \beta}}f(X) =
\begin{cases} 
      2 & \textrm{if }  \, \, \, (i,j)= (\alpha, \beta) \\
      0 & \textrm{otherwise}.
   \end{cases}
   \end{equation}
Define $X(t):=  \int_0^{t}\sum_{i=1}^{d} \sum_{j \neq i} |\lambda_i - \lambda_j| \frac{\mathrm{d}B_{ij}(s)}{\max(|\Delta_{ij}(s)|, \eta_{ij})}(u_i(s) u_j^\top(s) + u_j(s) u_i^\top(s))$ for all $t \geq 0$.
Then
\begin{equation*}
\mathrm{d}X_{\ell r}(t) = \sum_{j=1}^d R_{(\ell r) (i j)}(t) \mathrm{d}B_{(ij)}(t) \qquad \qquad \forall t \geq 0,
\end{equation*}
where $R_{(\ell r) (i j)}(t) := \left(\frac{ |\lambda_i - \lambda_j| }{\max(|\Delta_{ij}(t)|, \eta_{ij})}(u_i(t) u_j^\top(t) + u_j(t) u_i^\top(t)) \right)[\ell, r]$, and where we denote by either $H_{\ell r}$ or $H[\ell, r]$  the $(\ell, r)$'th entry of any matrix $H$.

Then we have
\begin{align}\label{eq_int_2b}
   &\mathbb{E}\left[\left\| \int_0^{T}\sum_{i=1}^{d} \sum_{j \neq i} |\lambda_i - \lambda_j| \frac{\mathrm{d}B_{ij}(t)}{\max(|\Delta_{ij}(t)|, \eta_{ij})}(u_i(t) u_j^\top(t) + u_j(t) u_i^\top(t))\right\|_F^2 \right] \nonumber \\
    &= \mathbb{E}[f(X(T))] \nonumber\\
   &=\mathbb{E}[f(X(T)) - f(X(0))] \nonumber\\
     &\stackrel{\textrm{It\^o's Lemma (Lemma \ref{lemma_ito_lemma_new})}}{=} \mathbb{E}\left [\frac{1}{2} \int_0^t \sum_{\ell, r} \sum_{\alpha, \beta} \left(\frac{\partial}{ \partial X_{\alpha \beta}} f(X(t))\right) R_{(\ell r) (\alpha \beta)}(t) \mathrm{d}B_{\ell r}(t) \right] \nonumber\\
     &\qquad \qquad \qquad \qquad \qquad +\mathbb{E}\left [\frac{1}{2} \int_0^t \sum_{\ell, r} \sum_{i,j} \sum_{\alpha, \beta} \left(\frac{\partial^2}{\partial X_{ij} \partial X_{\alpha \beta}} f(X(t))\right) R_{(\ell r) (i j)}(t)  R_{(\ell r) (\alpha \beta)}(t) \mathrm{d}t \right] \nonumber\\
        &= 0 +\mathbb{E}\left [\frac{1}{2} \int_0^t \sum_{\ell, r} \sum_{i,j} \sum_{\alpha, \beta} \left(\frac{\partial^2}{\partial X_{ij} \partial X_{\alpha \beta}} f(X(t))\right) R_{(\ell r) (i j)}(t)  R_{(\ell r) (\alpha \beta)}(t) \mathrm{d}t \right],
        \end{align}
        where the third equality is It\^o's Lemma (Lemma \ref{lemma_ito_lemma_new}), and the last equality holds since $$\mathbb{E}\left[\int_0^T  \left(\frac{\partial}{ \partial X_{\alpha \beta}} f(X(t))\right) R_{(\ell r) (\alpha \beta)}(t) \mathrm{d}B_{\ell r}(t) \right] = 0$$ for each $\ell, r, \alpha, \beta \in [d]$ because $\mathrm{d}B_{\ell r}(s)$ is independent of both $X(t)$ and $R(t)$ for all $s \geq t$ and the Brownian motion increments $\mathrm{d}B_{\alpha \beta}(s)$ satisfy $\mathbb{E}[\int_t^{\tau} \mathrm{d}B_{\alpha \beta}(s)] = \mathbb{E}[B_{\alpha \beta}(\tau) - B_{\alpha \beta}(t)]= 0$ for any $\tau \geq t$.

        Thus, plugging \eqref{eq_int_5} into \eqref{eq_int_2b}, we have 
   \begin{align}\label{eq_int_2}
       &\mathbb{E}\left[\left\| \int_0^{T}\sum_{i=1}^{d} \sum_{j \neq i} |\lambda_i - \lambda_j| \frac{\mathrm{d}B_{ij}(t)}{\max(|\Delta_{ij}(t)|, \eta_{ij})}(u_i(t) u_j^\top(t) + u_j(t) u_i^\top(t))\right\|_F^2 \right] \nonumber\\ 
   &\stackrel{\textrm{Eq.} \eqref{eq_int_5}, \eqref{eq_int_2b}}{=} \mathbb{E}\left [\frac{1}{2} \int_0^t \sum_{\ell, r} \sum_{i,j}  2  R_{(\ell r) (i j)}^2(t) \mathrm{d}t \right]\nonumber\\
                   &= \mathbb{E}\left [ \int_0^t \sum_{\ell, r} \sum_{i,j} \left(\left( \frac{ |\lambda_i - \lambda_j| }{\max(|\Delta_{ij}(t)|, \eta_{ij})^2}(u_i(t) u_j^\top(t) + u_j(t) u_i^\top(t))\right)[\ell, r]\right)^2 \mathrm{d}t \right]\nonumber\\
                                      &= \mathbb{E}\left [ \int_0^t \sum_{i,j} \sum_{\ell, r}  \left(\left( \frac{ |\lambda_i - \lambda_j| }{\max(|\Delta_{ij}(t)|, \eta_{ij})^2}(u_i(t) u_j^\top(t) + u_j(t) u_i^\top(t))\right)[\ell, r]\right)^2 \mathrm{d}t \right]\nonumber\\
                        &= \mathbb{E}\left [ \int_0^t   \sum_{i,j} \left \|\frac{ |\lambda_i - \lambda_j| }{\max(|\Delta_{ij}(t)|, \eta_{ij})}(u_i(t) u_j^\top(t) + u_j(t) u_i^\top(t))\right\|_F^2\mathrm{d}t \right]\nonumber\\
                &= 2\int_0^{T}   \mathbb{E}\left[ \sum_{i=1}^{d}  \sum_{j \neq i}  \frac{(\lambda_i - \lambda_j)^2}{\max(\Delta^2_{ij}(t), \eta_{ij}^2)} \|u_i(t) u_j^\top(t) + u_j(t) u_i^\top(t)\|_F^2  \mathrm{d}t\right] \nonumber\\
                             &= 4\int_0^{T}   \mathbb{E}\left[ \sum_{i=1}^{d}  \sum_{j \neq i}  \frac{(\lambda_i - \lambda_j)^2}{\max(\Delta^2_{ij}(t), \eta_{ij}^2)} \mathrm{d}t\right],
          \end{align}
where the fifth equality holds because $\langle u_i(t) u_j^\top(t) ,  u_\ell(t) u_h^\top(t) \rangle = 0$ for all $(i,j) \neq (\ell,h)$, and the last equality holds because  $\|u_i(t) u_j^\top(t) + u_j(t) u_i^\top(t)\|_F^2 = 2$ for all $t$ with probability $1$.

\paragraph{Bounding the drift term:}

To bound the drift term in \eqref{eq_int_1}, we use the Cauchy-Shwarz inequality:
\begin{align}\label{eq_int_3}
     &\left\|\int_0^{T}\sum_{i=1}^{d} \sum_{j\neq i} (\lambda_i - \lambda_j) \frac{\mathrm{d}t}{\max(\Delta^2_{ij}(t), \eta_{ij}^2)} u_i(t) u_i^\top(t) \right\|_F^2\nonumber \\
     & =      \left\|\int_0^{T}\sum_{i=1}^{d} \sum_{j\neq i}  \frac{\lambda_i - \lambda_j}{\max(\Delta^2_{ij}(t), \eta_{ij}^2)} u_i(t) u_i^\top(t) \times 1 \mathrm{d}t \right\|_F^2 \nonumber\\ 
     &     \stackrel{\textrm{Cauchy-Schwarz Inequality}}{\leq}     \int_0^{T}\left\|\sum_{i=1}^{d} \sum_{j\neq i}  \frac{\lambda_i - \lambda_j}{\max(\Delta^2_{ij}(t), \eta_{ij}^2)} u_i(t) u_i^\top(t)\right\|_F^2 \mathrm{d}t\times \int_0^{T} 1^2 \mathrm{d}t \nonumber\\
     &=   T \int_0^{T}\left\|\sum_{i=1}^{d} \sum_{j\neq i} \frac{\lambda_i - \lambda_j}{\max(\Delta^2_{ij}(t), \eta_{ij}^2)} u_i(t) u_i^\top(t) \right\|_F^2 \mathrm{d}t \nonumber \\
          &=   T \int_0^{T}\sum_{i=1}^{d} \left\|\sum_{j\neq i} \frac{\lambda_i - \lambda_j}{\max(\Delta^2_{ij}(t), \eta_{ij}^2)} u_i(t) u_i^\top(t) \right\|_F^2 \mathrm{d}t \nonumber\\
                    &=   T \int_0^{T}\sum_{i=1}^{d} \left\|\left(\sum_{j\neq i} \frac{\lambda_i - \lambda_j}{\max(\Delta^2_{ij}(t), \eta_{ij}^2)}\right) u_i(t) u_i^\top(t) \right\|_F^2 \mathrm{d}t \nonumber\\
                                        &=   T \int_0^{T}\sum_{i=1}^{d}\left(\sum_{j\neq i} \frac{\lambda_i - \lambda_j}{\max(\Delta^2_{ij}(t), \eta_{ij}^2)}\right)^2 \left\| u_i(t) u_i^\top(t) \right\|_F^2 \mathrm{d}t \nonumber\\
                                                                                &=   T \int_0^{T}\sum_{i=1}^{d}\left(\sum_{j\neq i} \frac{\lambda_i - \lambda_j}{\max(\Delta^2_{ij}(t), \eta_{ij}^2)}\right)^2 \times 1 \mathrm{d}t,
\end{align}
where the first inequality is by the Cauchy-Schwarz inequality for integrals (applied to each entry of the matrix-valued integral).
The third equality holds since $\langle u_i(t) u_i^\top(t) , u_j(t) u_j^\top(t) \rangle = 0$ for all $i \neq j$.
The last equality holds since $\| u_i(t) u_i^\top(t) \|_F^2=1$ with probability $1$.
Therefore, taking the expectation on both sides of \eqref{eq_int_1}, and plugging \eqref{eq_int_2} and  \eqref{eq_int_3} into  \eqref{eq_int_1}, we have
\begin{align} \label{eq_int_4}
\mathbb{E}\left[\left\|Z_\eta\left(T\right) -  Z_\eta(0)\right \|_F^2\right] &\leq  2\int_0^{T}   \mathbb{E}\left[ \sum_{i=1}^{d}  \sum_{j \neq i}  \frac{(\lambda_i - \lambda_j)^2}{\max(\Delta^2_{ij}(t), \eta_{ij}^2)} \right]\mathrm{d}t \nonumber \\
    & +    T \int_0^{T}\mathbb{E}\left[\sum_{i=1}^{d}\left(\sum_{j\neq i} \frac{\lambda_i - \lambda_j}{\max(\Delta^2_{ij}(t), \eta_{ij}^2)}\right)^2\right] \mathrm{d}t .
\end{align}
\end{proof}

\section{Proof of Theorem \ref{thm_large_gap} -- Main Result} \label{sec_proof_thm_large_gap}

\begin{proof}[Proof of Theorem \ref{thm_large_gap}]

To complete the proof of Theorem \ref{thm_large_gap}, we plug in the high-probability concentration bounds on the eigenvalue gaps $\Delta_{ij}(t) = \gamma_{i}(t) - \gamma_j(t)$ (Lemma \ref{lemma_gap_concentration}) into Lemma \ref{Lemma_integral}.
Since by Lemma \ref{lemma_gap_concentration} $\Delta_{ij}(t) \geq \frac{1}{2}(\sigma_i - \sigma_j)$ w.h.p. for each $i,j \leq k+1$, and $\eta_{ij} = \frac{1}{4}(\sigma_i - \max(\sigma_{j}, \sigma_{k+1}))$,  by Lemma \ref{Lemma_orbit_differntial} we have that the derivative $\mathrm{d}\Psi(t)$ satisfies $\mathrm{d}\Psi(t) = \mathrm{d}Z_\eta(t)$ for all $t \in [0,T]$ w.h.p. and hence that $Z_\eta(t) = \Psi(t)$ for all $t \in [0, T]$ w.h.p.
Plugging in the high-probability bounds on the gaps $\Delta_{ij}(t)$ (Lemma \ref{lemma_gap_concentration}) into the bound on $\mathbb{E}\left[\left\|Z_\eta\left(T\right) -  Z_\eta(0)\right \|_F^2\right]$ from Lemma \ref{Lemma_integral} therefore allows us to obtain a bound for $\mathbb{E}\left[\left\|\Psi\left(T\right) -  \Psi(0)\right \|_F^2\right]$.
To obtain a bound for the utility $\mathbb{E}\left[\| \hat{V} \Lambda \hat{V}^\top -  V \Lambda V^\top \|_F^2\right]$ we set $T = \frac{\sqrt{2\log(\frac{1.25}{\delta})}}{\eps}$, in which case we have $\Psi(T) =\hat{V} \Lambda \hat{V}^\top$ and hence that $\mathbb{E}\left[\| \hat{V} \Lambda \hat{V}^\top -  V \Lambda V^\top \|_F^2\right] = \mathbb{E}\left[\left\|\Psi\left(T\right) -  \Psi(0)\right \|_F^2\right]$.

Towards this end, for all $i\neq j$ we define $\eta_{ij}$ as follows:
 Let $\eta_{ij} = \frac{1}{4}(\sigma_i - \max(\sigma_{j}, \sigma_{k+1}))$ for $0<i<j \leq d$ and $i\leq k$, $\eta_{ij}= 0$ if $0<i<j \leq d$ and $i> k$, and $\eta_{ij} = \eta_{ji}$ otherwise. 

Define the event $E = \cap_{i,j \in [d], i\neq j}\{ \inf_{t\in[0,T]}\Delta_{ij}(t) \geq \eta_{ij}\}$. 
And define the event $$\hat{E}:=\bigcap_{i\in [k]} \left\{\inf_{t \in [0,T]} \gamma_i(t) - \gamma_{i+1}(t) \geq  \frac{1}{4}(\sigma_i - \sigma_{i+1})) \right\}.$$
Then $\hat{E} \subseteq E$. 
In particular, whenever the event $E$ occurs, by Lemma \ref{Lemma_orbit_differntial} we have that the derivative $\mathrm{d}\Psi(t)$ satisfies $\mathrm{d}\Psi(t) = \mathrm{d}Z_\eta(t)$ for all $t \in [0,T]$ and hence that 
\begin{equation*}
\Psi(t) = \Psi(0) + \int_0^t \mathrm{d}\Psi(s) = Z_\eta(0) + \int_0^t \mathrm{d}Z_\eta(s) =  Z_\eta(t)\qquad \qquad \forall t \in [0,T],
\end{equation*}
whenever the event $E$ occurs, since $Z_\eta(0) = \Psi(0)$ by definition.
Thus we have that, conditioning $\Psi(t)$ and $Z_\eta(t)$ on the event $E$,
\begin{equation}\label{eq_e6}
    \Psi(t)|E = Z_\eta(t)|E \qquad \forall t \in [0,T].
\end{equation}
To bound the utility $\mathbb{E}\left[\| \hat{V} \Lambda \hat{V}^\top -  V \Lambda V^\top \|_F^2\right]$, we first separate $\mathbb{E}\left[\| \hat{V} \Lambda \hat{V}^\top -  V \Lambda V^\top \|_F^2\right]$ into a sum of terms conditioned on the event $E$ and its complement $E^c$.
By Lemma \ref{Lemma_integral} we have
\begin{align}\label{eq_e1}
    &\mathbb{E}\left[\| \hat{V} \Lambda \hat{V}^\top -  V \Lambda V^\top \|_F^2\right] \nonumber \\
    &= \mathbb{E}\left[\| \hat{V} \Lambda \hat{V}^\top -  V \Lambda V^\top \|_F^2\bigg| \, E\right]\times \mathbb{P}(E) + \mathbb{E}\left[\| \hat{V} \Lambda \hat{V}^\top -  V \Lambda V^\top \|_F^2\bigg| E^c \, \right]\times \mathbb{P}(E^c)\nonumber \\
        &\leq \mathbb{E}\left[\| \hat{V} \Lambda \hat{V}^\top -  V \Lambda V^\top \|_F^2\bigg| \, E\right]\times \mathbb{P}(E) \nonumber\\
        &+ 4\left(\mathbb{E}\left[\| (\hat{V} - V)\Lambda \hat{V}^\top\|_F^2\bigg| E^c \, \right] +\mathbb{E}\left[\| V \Lambda (\hat{V}-V)^\top \|_F^2\bigg| E^c \, \right]\right)\times \mathbb{P}(E^c)\nonumber \\
                &\leq \mathbb{E}\left[\| \hat{V} \Lambda \hat{V}^\top -  V \Lambda V^\top \|_F^2|E\right]\times \mathbb{P}(E)  + 8\mathbb{E}\left[\|\hat{V} -V\|_2^2\times \|\Lambda\|_F^2 \, \bigg| \, E^c\right]\times \mathbb{P}(E^c)\nonumber \\
                                &\leq \mathbb{E}\left[\| \hat{V} \Lambda \hat{V}^\top -  V \Lambda V^\top \|_F^2\bigg| \, E\right]\times \mathbb{P}(E)  + 32\|\Lambda\|_F^2\times \mathbb{P}(E^c) \nonumber \\
    &=\mathbb{E}\left[\left\|\Psi\left(T\right) -  \Psi(0)\right \|_F^2 \bigg| \, E\right] \times \mathbb{P}(E) + 32\|\Lambda\|_F^2\times \mathbb{P}(E^c) \nonumber \\
    &\leq\mathbb{E}\left[\left\|\Psi\left(T\right) -  \Psi(0)\right \|_F^2 \bigg| \, E\right] \times \mathbb{P}(E)  + 32\lambda_1^2 k\times \mathbb{P}(\hat{E}^c),
    \end{align}
    where the second inequality holds by the sub-multiplicative property of the Frobenius norm which says that $\|XY\|_F \leq \|X\|_2\times \|Y\|_F$ for any $X,Y \in \mathbb{R}^{d \times d}$.
    The third inequality holds since $\|V\|_2 = \|\hat{V}\|_2= 1$ since $V, \hat{V}$ are orthogonal matrices.

    To bound the first term in \eqref{eq_e1}, we use the fact that $\Psi(t)|E = Z_\eta(t)|E$ (Equation \eqref{eq_e6}) and apply Lemma \ref{Lemma_integral} to bound  $\mathbb{E}\left[\left\|Z_\eta\left(T\right) -  Z_\eta(0)\right \|_F^2\right]$.
    Thus we have,
    \begin{align} \label{eq_e2}
    &\mathbb{E}\left[\left\|\Psi\left(T\right) -  \Psi(0)\right \|_F^2 \bigg| \, E\right]\times \mathbb{P}(E) \nonumber\\
    &\stackrel{\textrm{Eq. }\eqref{eq_e6}}{=} \mathbb{E}\left[\left\|Z_\eta\left(T\right) -  Z_\eta(0)\right \|_F^2 \bigg| \, E\right]\times \mathbb{P}(E) \nonumber \\
    &\leq \mathbb{E}\left[\left\|Z_\eta\left(T\right) -  Z_\eta(0)\right \|_F^2\right] \nonumber \\
    &\stackrel{\textrm{Lemma }\ref{Lemma_integral}}{=}  2\int_0^{T}   \mathbb{E}\left[ \sum_{i=1}^{d}  \sum_{j \neq i}  \frac{(\lambda_i - \lambda_j)^2}{\max(\Delta^2_{ij}(t), \eta_{ij}^2)}\right] \mathrm{d}t
     +    T \int_0^{T}\mathbb{E}\left[\sum_{i=1}^{d}\left(\sum_{j\neq i} \frac{\lambda_i - \lambda_j}{\max(\Delta^2_{ij}(t), \eta_{ij}^2)}\right)^2  \right]\mathrm{d}t \nonumber \\
     &\leq 4\int_0^{T}   \mathbb{E}\left[ \sum_{i=1}^{d}  \sum_{j = i+1}^d  \frac{(\lambda_i - \lambda_j)^2}{\max(\Delta^2_{ij}(t), \eta_{ij}^2)}\right] \mathrm{d}t
     +    2T \int_0^{T}\mathbb{E}\left[\sum_{i=1}^{d}\left(\sum_{j=i+1}^d \frac{|\lambda_i - \lambda_j|}{\max(\Delta^2_{ij}(t), \eta_{ij}^2)}\right)^2  \right]\mathrm{d}t \nonumber \\
          &=4\int_0^{T}   \mathbb{E}\left[ \sum_{i=1}^{k}  \sum_{j = i+1}^d  \frac{(\lambda_i - \lambda_j)^2}{\max(\Delta^2_{ij}(t), \eta_{ij}^2)}\right] \mathrm{d}t
     +    2T \int_0^{T}\mathbb{E}\left[\sum_{i=1}^{k}\left(\sum_{j=i+1}^d \frac{|\lambda_i - \lambda_j|}{\max(\Delta^2_{ij}(t), \eta_{ij}^2)}\right)^2  \right]\mathrm{d}t \nonumber\\
     &\leq 64\int_0^{T}   \mathbb{E}\left[ \sum_{i=1}^{k}  \sum_{j = i+1}^d  \frac{(\lambda_i - \lambda_j)^2}{(\sigma_i-\max(\sigma_j, \sigma_{k+1}))^2} \right] \mathrm{d}t \nonumber \\
    &  +    32T \int_0^{T}\mathbb{E}\left[\sum_{i=1}^{k}\left(\sum_{j = i+1}^d \frac{|\lambda_i - \lambda_j|}{(\sigma_i-\max(\sigma_j, \sigma_{k+1}))^2}\right)^2  \right]\mathrm{d}t \nonumber \\
          &= 64T   \mathbb{E}\left[ \sum_{i=1}^{k}  \sum_{j = i+1}^d  \frac{(\lambda_i - \lambda_j)^2}{(\sigma_i-\max(\sigma_j, \sigma_{k+1}))^2} \right]
     +    32T^2 \mathbb{E}\left[\sum_{i=1}^{k}\left(\sum_{j=i+1}^d \frac{|\lambda_i - \lambda_j|}{(\sigma_i-\max(\sigma_j, \sigma_{k+1}))^2}\right)^2  \right] \nonumber \\
               &= 64T   \sum_{i=1}^{k}  \sum_{j = i+1}^d  \frac{(\lambda_i - \lambda_j)^2}{(\sigma_i-\max(\sigma_j, \sigma_{k+1}))^2}
     +    32T^2 \sum_{i=1}^{k}\left(\sum_{j=i+1}^d \frac{|\lambda_i - \lambda_j|}{(\sigma_i-\max(\sigma_j, \sigma_{k+1}))^2}\right)^2,
\end{align}
where the first equality holds since $\Psi(t)|E = Z_\eta(t)|E$ by \eqref{eq_e6}, and the second equality holds by Lemma  \ref{Lemma_integral}.
 The second inequality holds since, $ \frac{(\lambda_i - \lambda_j)^2}{\max(\Delta^2_{ij}(t), \eta_{ij}^2)} =  \frac{(\lambda_j - \lambda_i)^2}{\max(\Delta^2_{ji}(t), \eta_{ji}^2)}$  and $\frac{|\lambda_i - \lambda_j|}{\max(\Delta^2_{ij}(t), \eta_{ij}^2)} = \frac{|\lambda_j - \lambda_i|}{\max(\Delta^2_{ji}(t), \eta_{ji}^2)}$ for all $i,j \in [d]$.
 The third equality holds since $\lambda_{i} = 0$ for all $i \geq k+1$.
 The third inequality holds since $\eta_{ij} = \frac{1}{4}(\sigma_i - \max(\sigma_{j}, \sigma_{k+1}))$ for all $0<i<j \leq d$.

To bound the second term in \eqref{eq_e1}, we have by Lemma \ref{lemma_gap_concentration} that
\begin{align}\label{eq_e3}
    \mathbb{P}(\hat{E}^c) &= \mathbb{P}\left(\left(\bigcap_{i\in [k]} \left\{\inf_{t \in [0,T]} \gamma_i(t) - \gamma_{i+1}(t) \geq  \frac{1}{4}(\sigma_i - \sigma_{i+1})) \right\}\right)^c \, \, \right) \nonumber \\
    &= \mathbb{P}\left(\bigcup_{i\in [k]} \left\{\inf_{t \in [0,T]} \gamma_i(t) - \gamma_{i+1}(t) <  \frac{1}{4}(\sigma_i - \sigma_{i+1})) \right\}\right) \nonumber\\
        &\stackrel{\textrm{Assumption} \ref{assumption_gaps}}{\leq} \mathbb{P}\left(\bigcup_{i\in [k]} \left\{\inf_{t \in [0,T]} \gamma_i(t) - \gamma_{i+1}(t) <  \frac{1}{2}(\sigma_i - \sigma_{i+1}) - 3\log^{\frac{1}{2}}(\lambda_1 k)) \right\}\right) \nonumber\\
        &\stackrel{\textrm{Lemma}  \ref{lemma_gap_concentration}}{\leq} \min(e^{-\log(\lambda_1^2 k))},1) \nonumber\\
        &=\min\left(\frac{1}{\lambda_1^2 k }, 1\right),
\end{align}
where the first inequality holds by Assumption \ref{assumption_gaps},
and the second inequality holds by Lemma \ref{lemma_gap_concentration}.

Therefore, plugging \eqref{eq_e2} and \eqref{eq_e3} into \eqref{eq_e1}, we have
\begin{align}\label{eq_e4}
    &\mathbb{E}\left[\| \hat{V} \Lambda \hat{V}^\top -  V \Lambda V^\top \|_F^2\right] \nonumber\\
    &\stackrel{\textrm{Eq. }\eqref{eq_e1}, \eqref{eq_e2}, \eqref{eq_e3}}{\leq} 64T   \sum_{i=1}^{k}  \sum_{i = j+1}^d  \frac{(\lambda_i - \lambda_j)^2}{(\sigma_i-\max(\sigma_j, \sigma_{k+1}))^2}
     +    32T^2 \sum_{i=1}^{k}\left(\sum_{i=j+1}^d \frac{|\lambda_i - \lambda_j|}{(\sigma_i-\max(\sigma_j, \sigma_{k+1}))^2}\right)^2 \nonumber\\ & + \min(32, \, \, 32\lambda_1^2 k) \nonumber\\
     &\leq       O\left(\sum_{i=1}^{k}  \sum_{j = i+1}^d  \frac{(\lambda_i - \lambda_j)^2}{(\sigma_i-\max(\sigma_j, \sigma_{k+1}))^2}
     +    \left(\sum_{j=i+1}^d \frac{|\lambda_i - \lambda_j|}{(\sigma_i-\max(\sigma_j, \sigma_{k+1}))^2}\right)^2\right) \frac{\log(\frac{1}{\delta})}{\eps^2},
\end{align}
where the last inequality holds since $T = \frac{\sqrt{2\log(\frac{1.25}{\delta})}}{\eps}$.

Finally, we have by the Cauchy-Schwarz inequality and Assumption \ref{assumption_gaps} that 
\begin{align}\label{eq_e5}
    &\left(\sum_{j=i+1}^d \frac{|\lambda_i - \lambda_j|}{(\sigma_i-\max(\sigma_j, \sigma_{k+1}))^2}\right)^2  =  \left(\sum_{j=i+1}^d \frac{1}{|\sigma_i-\max(\sigma_j, \sigma_{k+1})|} \times  \frac{|\lambda_i - \lambda_j|}{|\sigma_i-\max(\sigma_j, \sigma_{k+1})|}\right)^2 \nonumber\\
    & \stackrel{\textrm{Cauchy-Schwarz inequality}}{\leq} \left(\sum_{j=i+1}^d \frac{1}{(\sigma_i-\max(\sigma_j, \sigma_{k+1}))^2}\right) \times \left(\sum_{j=i+1}^d \frac{(\lambda_i - \lambda_j)^2}{(\sigma_i-\max(\sigma_j, \sigma_{k+1}))^2}\right)\nonumber\\
    &\stackrel{\textrm{Assumption} \ref{assumption_gaps}}{\leq} 
   \left(\sum_{j=i+1}^d \frac{1}{(\sqrt{d})^2}\right) \times \left(\sum_{j=i+1}^d \frac{(\lambda_i - \lambda_j)^2}{(\sigma_i-\max(\sigma_j, \sigma_{k+1}))^2}\right)\nonumber\\
   &\leq \sum_{j=i+1}^d \frac{(\lambda_i - \lambda_j)^2}{(\sigma_i-\max(\sigma_j, \sigma_{k+1}))^2}.
\end{align}
In other words, \eqref{eq_e5} says that the first term inside the outer summation on the r.h.s. of \eqref{eq_e4} is at least as large as the second term.
Therefore, plugging in \eqref{eq_e5} into \eqref{eq_e4}, we have that
\begin{align*}
    \mathbb{E}\left[\| \hat{V} \Lambda \hat{V}^\top -  V \Lambda V^\top \|_F^2\right] \leq      O\left(\sum_{i=1}^{k}  \sum_{j = i+1}^d  \frac{(\lambda_i - \lambda_j)^2}{(\sigma_i-\max(\sigma_j, \sigma_{k+1}))^2}\right) \frac{\log(\frac{1}{\delta})}{\eps^2}.
    \end{align*}
\end{proof}

\section{Proof of Corollary \ref{cor_rank_k_covariance2} -- Covariance Matrix  Approximation}\label{sec:covariance}

\begin{proof}[Proof of Corollary \ref{cor_rank_k_covariance2}]

 To prove Corollary \ref{cor_rank_k_covariance2},
 we must bound the utility $\mathbb{E}[\| \hat{V} \hat{\Sigma}_k \hat{V}^\top -  V \Sigma_k V^\top \|_F]$ of the post-processing of the Gaussian mechanism for the rank-$k$ covariance matrix estimation problem.
Towards this end, we first plug in $\lambda_i = \sigma_i$ for $i \leq k$ and $\lambda_i=0$ for $i>k$ into Theorem \ref{thm_large_gap} to obtain a bound for $\mathbb{E}[\| \hat{V} \Sigma_k \hat{V}^\top -  V \Sigma_k V^\top \|_F]$ (Inequality \eqref{eq_f6}).
 We then apply Weyl's inequality (Lemma \ref{lemma_weyl}) together with a concentration bound for $\|B(T)\|_2$  (Lemma \ref{lemma_spectral_martingale}) to bound the perturbation to the eigenvalues of $M$ when the Gaussian noise matrix $B(T)$ is added to $M$ by the Gaussian mechanism.
 This implies a bound on  $\mathbb{E}[\| \hat{V} \Sigma_k \hat{V}^\top -  \hat{V} \hat{\Sigma}_k \hat{V}^\top \|_F]$ (Inequality \eqref{eq_f8}).
 Combining these two bounds \eqref{eq_f6} and \eqref{eq_f8}, implies a bound on the utility $\mathbb{E}[\| \hat{V} \hat{\Sigma}_k \hat{V}^\top -  V \Sigma_k V^\top \|_F]$  for the post-processing of the Gaussian mechanism (Inequality \eqref{eq_f11}).

\paragraph{Bounding the quantity $\mathbb{E}[\| \hat{V} \Sigma_k \hat{V}^\top -  V \Sigma_k V^\top \|_F]$.}

Let $\lambda_i = \sigma_i$ for $i \leq k$ and $\lambda_i=0$ for $i>k$ , and let $\Lambda:= \mathrm{diag}(\lambda_1,\cdots, \lambda_d)$.
Then by Assumption \ref{assumption_gaps} we have that $\sigma_i - \sigma_{i+1} \geq  \frac{8\sqrt{\log(\frac{1.25}{\delta})}}{\eps} \sqrt{d} + c\log^{\frac{1}{2}}(\sigma_1 k)$.
By Theorem \ref{thm_large_gap}, we have

\begin{equation}\label{eq_f4}
    \mathbb{E}\left[\| \hat{V} \Lambda \hat{V}^\top -  V \Lambda V^\top \|_F^2\right] 
\leq       O\left(\sum_{i=1}^{k}  \sum_{j = i+1}^d  \frac{(\lambda_i - \lambda_j)^2}{(\sigma_i-\max(\sigma_j, \sigma_{k+1}))^2}     \right) \frac{\log(\frac{1}{\delta})}{\eps^2}.
\end{equation}
First, we note that $\lambda_i - \lambda_j \leq \sigma_i - \sigma_j$ for all $i \leq j \leq k$.
Then for all $i<j\leq k$, we have
\begin{align*}
\frac{\lambda_i - \lambda_j}{\sigma_i - \max(\sigma_j, \sigma_{k+1})}  \leq
\frac{\sigma_i - \sigma_k}{\sigma_i - \max(\sigma_j, \sigma_{k+1})} \leq 1.
\end{align*}
And for all $i\leq k <j\leq d$ we have
\begin{align*}
&\frac{\lambda_i - \lambda_j}{\sigma_i - \max(\sigma_j, \sigma_{k+1})}  =
\frac{\sigma_i}{\sigma_i -  \sigma_{k+1}} = \frac{\sigma_i-\sigma_k}{\sigma_i - \sigma_{k+1}} + \frac{\sigma_k}{\sigma_i - \sigma_{k+1}}
\leq 1 + \frac{\sigma_k}{\sigma_k - \sigma_{k+1}}.
\end{align*}
Thus, the summation term in \eqref{eq_f4} simplifies to 
\begin{equation} \label{eq_f1}
\sum_{i=1}^{k}  \sum_{j = i+1}^d  \frac{(\lambda_i - \lambda_j)^2}{(\sigma_i-\max(\sigma_j, \sigma_{k+1}))^2}
\leq kd\left( 1 + \frac{\sigma_k}{\sigma_k - \sigma_{k+1}}\right)^2 \leq 4kd\left(\frac{\sigma_k}{\sigma_k - \sigma_{k+1}}\right)^2.
\end{equation}
 Therefore, plugging \eqref{eq_f1} 
 into \eqref{eq_f4}, we have
 \begin{align}\label{eq_f6}
      \mathbb{E}[\| \hat{V} \Sigma_k \hat{V}^\top -  V \Sigma_k V^\top \|_F] &\leq \sqrt{\mathbb{E}[\| \hat{V} \Sigma_k \hat{V}^\top -  V \Sigma_k V^\top \|_F^2]} \nonumber\\
      & \stackrel{\textrm{Eq.} \eqref{eq_f1}, \eqref{eq_f4} }{\leq} O\left(\sqrt{kd} \times \frac{\sigma_k}{\sigma_k - \sigma_{k+1}}\right) \frac{\log^{\frac{1}{2}}(\frac{1}{\delta})}{\eps},
 \end{align}
 where the first inequality holds by Jensen's inequality.

\paragraph{Bounding the perturbation to the eigenvalues.}

     By Weyl's inequality (Lemma \ref{lemma_weyl}), we have that for every $i \in [d]$
\begin{align}\label{eq_f9}
   \mathbb{E}[(\hat{\sigma}_i-\sigma_i)^2] &\leq \mathbb{E}[(\|B(T)\|_2)^2] \nonumber \\
   &\leq 4\mathbb{E}[( (T\sqrt{d})^2] + 4\mathbb{E}\left[\left(\|B(T)\|_2 - T\sqrt{d} \right)^2\right] \nonumber\\
   &\leq   4T^2d + 4\int_0^\infty \mathbb{P}\left(\left(\|B(T)\|_2 - T\sqrt{d} \right)^2 >  \alpha)\right) \mathrm{d} \alpha \nonumber \\
      &=   4T^2d + 4\int_0^\infty \mathbb{P}\left(\|B(T)\|_2 - T\sqrt{d}  >  \sqrt{\alpha})\right) \mathrm{d} \alpha \nonumber\\
      &\stackrel{\textrm{Lemma} \eqref{lemma_spectral_martingale} }{\leq} 4T^2d + 8\sqrt{\pi} \int_0^\infty e^{-\frac{1}{8}\frac{\alpha}{T^2}}\mathrm{d} \alpha \nonumber\\
      &= 4T^2d + 64\sqrt{\pi} T^2 e^{-\frac{1}{8}\frac{\alpha}{T^2}} \bigg|_{\alpha=0}^\infty \nonumber\\
      &=  4T^2d + 64\sqrt{\pi} T^2\nonumber\\
      &\leq 64\sqrt{\pi}  \frac{\log(\frac{1}{\delta})}{\eps^2} d.
\end{align}
 The first inequality holds by Weyl's inequality (Lemma \ref{lemma_weyl}), and the fourth inequality holds by Lemma \ref{lemma_spectral_martingale}.
Therefore, \eqref{eq_f9} implies that,
\begin{align}\label{eq_f8}
 \mathbb{E}[\| \hat{V} \Sigma_k \hat{V}^\top -  \hat{V} \hat{\Sigma}_k \hat{V}^\top \|_F]  &=  \mathbb{E}[\|  \Sigma_k-  \hat{\Sigma}_k  \|_F] \nonumber \\
 &\leq  \sqrt{\mathbb{E}[\|  \Sigma_k-  \hat{\Sigma}_k  \|_F^2]} \nonumber \\
 &\stackrel{\textrm{Eq.} \eqref{eq_f9} }{\leq}  O\left(\sqrt{kd}\frac{\log^{\frac{1}{2}}(\frac{1}{\delta})}{\eps}\right),
\end{align}
where the first inequality holds by Jensen's inequality, and the second inequality holds by \eqref{eq_f9}.
Thus, plugging \eqref{eq_f8} into \eqref{eq_f6}, we have that
\begin{align}\label{eq_f11}
 \mathbb{E}[\| \hat{V} \hat{\Sigma}_k \hat{V}^\top -  V \Sigma_k V^\top \|_F]  &\leq     \mathbb{E}[ \| \hat{V} \Sigma_k \hat{V}^\top -  V \Sigma_k V^\top \|_F] +   \mathbb{E}[\| \hat{V} \Sigma_k \hat{V}^\top -  \hat{V} \hat{\Sigma}_k \hat{V}^\top \|_F] \nonumber\\
 &\stackrel{\textrm{Eq.} \eqref{eq_f6}, \eqref{eq_f8} }{\leq} O\left(\sqrt{kd} \times\frac{\sigma_k}{\sigma_{k}-\sigma_{k+1}
    }\right) \frac{\log^{\frac{1}{2}}(\frac{1}{\delta})}{\eps}.
\end{align}

\paragraph{Privacy:}

{\em Privacy of perturbed covariance matrix $\hat{M}$:} 
Recall that two matrices $M=A^\top A$ and $M' = A'^\top A'$ are said to be {\em neighbors} if they arise from $A, A' \in \mathbb{R}^{d\times n}$ which differ by at most one row, and that each row of the datasets $A, A'$ has norm at most $1$.
In other words, we have that $M-M' = xx^\top$ for some $x\in \mathbb{R}^d$ such that $\|x\| \leq 1$.
Define the sensitivity $S:= \max_{M,M' \mathrm{neighbors}} \|M - M'\|_{\ell_2}$, where $\|X\|_{\ell_2}$ denotes the Euclidean norm of the upper triangular entries of $X$ (including the diagonal entries).
Then we have
\begin{equation*}
    S = \max_{M,M' \mathrm{neighbors}} \|M - M'\|_{\ell_2} \leq  \|M - M'\|_F \leq \max_{\|x\|\leq 1} \|xx^\top \|_F = 1.
\end{equation*}
Then by standard results for the Gaussian Mechanism (e.g., by Theorem A.1 of \cite{dwork2014algorithmic}), we have that the Gaussian mechanism which outputs the upper triangular matrix $\hat{M}_{\mathrm{upper}}$ with the same upper triangular entries as $\hat{M} = M+  \frac{S\sqrt{2\log(\frac{1.25}{\delta})}}{\eps}(G +G^\top) = M+ \frac{\sqrt{2\log(\frac{1.25}{\delta})}}{\eps}(G +G^\top)$, where $G$ has i.i.d. $N(0,1)$ entries, is $(\varepsilon, \delta)$-differentially private.
However, since the perturbed matrix $\hat{M}$ is symmetric, it can be obtained from its upper triangular entries $\hat{M}_{\mathrm{upper}}$ without accessing the original matrix $M$.
Thus, the mechanism which outputs $\hat{M} = M+ \frac{\sqrt{2\log(\frac{1.25}{\delta})}}{\eps}(G +G^\top)$ must also be $(\varepsilon, \delta)$-differentially private.

{\em Privacy of rank-$k$ approximation $ \hat{V} \hat{\Sigma}_k \hat{V}^\top$:}    The mechanism which outputs the rank-$k$ approximation $\hat{V} \hat{\Sigma}_k \hat{V}^\top$ is $(\varepsilon, \delta)$-differentially private, since  $\hat{V} \hat{\Sigma}_k \hat{V}^\top$ is obtained by post-processing the perturbed matrix $\hat{M}$ without any additional access to the matrix $M$.

Namely, to obtain  $\hat{V} \hat{\Sigma}_k \hat{V}^\top$, we first (i) compute the spectral decomposition $\hat{M}= \hat{V} \hat{\Sigma} \hat{V}^\top$.
Next, (ii) we take the top-$k$ eigenvalues $ \hat{\sigma}_1, \ldots, \hat{\sigma}_k$ of  $\hat{M}$, and set $\hat{\Sigma}_k = \mathrm{diag}(\hat{\sigma}_1, \ldots, \hat{\sigma}_k, 0, \ldots, 0)$.
Finally, we output $\hat{M}_k := \hat{V}\hat{\Sigma}_k \hat{V}^\top$.
Both of these steps (i) and (ii) are post-processing of $\hat{M}$ and do not require additional access to the matrix $M$.
In particular, the eigenvalues  $ \hat{\sigma}_1, \ldots, \hat{\sigma}_k$ of  $\hat{M}_k$ are obtained from the perturbed matrix $\hat{M}$, and thus do not compromise privacy.
Therefore, the mechanism which outputs the rank-$k$ approximation $\hat{M}_k :=\hat{V} \hat{\Sigma} \hat{V}^\top$ must also be $(\varepsilon, \delta)$-differentially private.

\end{proof}

\section{Proof of Corollary \ref{cor_subspace_recovery} -- Subspace Recovery}\label{sec:cor_rank_subspace}

\begin{proof}[Proof of Corollary \ref{cor_subspace_recovery}]

To prove Corollary  \ref{cor_subspace_recovery}, we  plug in $\lambda_1 = \cdots = \lambda_k =1$ and $\lambda_{k+1} = \cdots = \lambda_d = 0$ to Theorem \ref{thm_large_gap}.
Corollary \ref{cor_subspace_recovery} considers two cases.
In the first case (referred to here as Case I), the   eigenvalues $\sigma$ of the input matrix $M$ satisfies Assumption  \ref{assumption_gaps}
In the second case (referred to here as Case II) the eigenvalues $\sigma$ of $M$ also satisfy both Assumption \ref{assumption_gaps}  as well as the lower bound $\sigma_i - \sigma_{i+1} \geq \Omega(\sigma_k - \sigma_{k+1})$ for all $i \leq k$.
We derive a bound on the utility $\mathbb{E}[\| \hat{V}_k \hat{V}_k^\top -  V_k  V_k^\top \|_F]$ in  each case separately.

\paragraph{Case I: $M$ satisfies Assumption  \ref{assumption_gaps}.}

Plugging in $\lambda_1 = \cdots = \lambda_k =1$ and $\lambda_{k+1} = \cdots = \lambda_d = 0$ to Theorem \ref{thm_large_gap} we get that, since $M$ satisfies Assumption \ref{assumption_gaps} for ($M,k,2,\eps,\delta$), 
\begin{align}\label{eq_h1}
 \mathbb{E}\left[\| \hat{V}_k \hat{V}_k^\top -  V_k V_k^\top \|_F^2\right] &= \mathbb{E}\left[\| \hat{V} \Lambda \hat{V}^\top -  V \Lambda V^\top \|_F^2\right]\nonumber\\
&\stackrel{\textrm{Theorem }\ref{thm_large_gap}}{\leq}         O\left(\sum_{i=1}^{k}  \sum_{j = i+1}^d  \frac{(\lambda_i - \lambda_j)^2}{(\sigma_i-\max(\sigma_j, \sigma_{k+1}))^2}
     \right) \frac{\log(\frac{1}{\delta})}{\eps^2} \nonumber\\
     &=  O\left(\sum_{i=1}^{k}  \sum_{j = k+1}^d  \frac{1}{(\sigma_i-\max(\sigma_j, \sigma_{k+1}))^2}     \right) \frac{\log(\frac{1}{\delta})}{\eps^2}\nonumber\\
  &\leq   O\left(\sum_{i=1}^{k}  \sum_{j = k+1}^d  \frac{1}{(\sigma_k- \sigma_{k+1})^2}     \right) \frac{\log(\frac{1}{\delta})}{\eps^2}\nonumber\\
 &= O\left(\frac{k d}{(\sigma_k- \sigma_{k+1})^2}      \frac{\log(\frac{1}{\delta})}{\eps^2}\right),
\end{align}
where the first inequality holds by Theorem \ref{thm_large_gap}, and the second equality holds since $\lambda_1 = \cdots = \lambda_k =1$ and $\lambda_{k+1} = \cdots = \lambda_d = 0$.

Thus, applying Jensen's Inequality to Inequality \eqref{eq_h1}, we have that
\begin{equation*}
\mathbb{E}[\| \hat{V}_k \hat{V}_k^\top -  V_k  V_k^\top \|_F] \leq  O\left(\frac{\sqrt{kd}}{(\sigma_k- \sigma_{k+1})}      \frac{\log^{\frac{1}{2}}(\frac{1}{\delta})}{\eps}\right).
\end{equation*}


\paragraph{Case II: $M$ satisfies Assumption \ref{assumption_gaps} and $\sigma_i - \sigma_{i+1} \geq \Omega(\sigma_k - \sigma_{k+1})$ for all $i \leq k$.} 

Plugging in $\lambda_1 = \cdots = \lambda_k =1$ and $\lambda_{k+1} = \cdots = \lambda_d = 0$ to Theorem \ref{thm_large_gap} we get that, since $M$ satisfies Assumption \ref{assumption_gaps} for ($M,k,2,\eps,\delta$), 
\begin{align}\label{eq_h2}
 \mathbb{E}\left[\| \hat{V}_k \hat{V}_k^\top -  V_k V_k^\top \|_F^2\right] &= \mathbb{E}\left[\| \hat{V} \Lambda \hat{V}^\top -  V \Lambda V^\top \|_F^2\right]\nonumber\\
&\stackrel{\textrm{Theorem }\ref{thm_large_gap}}{\leq}         O\left(\sum_{i=1}^{k}  \sum_{j = i+1}^d  \frac{(\lambda_i - \lambda_j)^2}{(\sigma_i-\max(\sigma_j, \sigma_{k+1}))^2}
     \right) \frac{\log(\frac{1}{\delta})}{\eps^2} \nonumber\\
     &=  O\left(\sum_{i=1}^{k}  \sum_{j = k+1}^d  \frac{1}{(\sigma_i-\max(\sigma_j, \sigma_{k+1}))^2}     \right) \frac{\log(\frac{1}{\delta})}{\eps^2}\nonumber\\
  &\leq   O\left(\sum_{i=1}^{k}  \sum_{j = k+1}^d  \frac{1}{(i-k-1)^2(\sigma_k- \sigma_{k+1})^2}     \right) \frac{\log(\frac{1}{\delta})}{\eps^2}\nonumber\\
  &\leq  O\left(\sum_{i=1}^{k}   \frac{d}{(i-k-1)^2(\sigma_k- \sigma_{k+1})^2}     \right) \frac{\log(\frac{1}{\delta})}{\eps^2}\nonumber\\
 &\leq O\left(\frac{d}{(\sigma_k- \sigma_{k+1})^2}      \frac{\log(\frac{1}{\delta})}{\eps^2} \sum_{i=1}^{k} \frac{1}{i^2} \right)\nonumber\\
 &\leq O\left(\frac{d}{(\sigma_k- \sigma_{k+1})^2}      \frac{\log(\frac{1}{\delta})}{\eps^2}\right),
\end{align}
where the first inequality holds by Theorem \ref{thm_large_gap}, the second equality holds since $\lambda_1 = \cdots = \lambda_k =1$ and $\lambda_{k+1} = \cdots = \lambda_d = 0$,
the second inequality holds since $\sigma_i - \sigma_{i+1} \geq \Omega(\sigma_k - \sigma_{k+1})$ for all $i \leq k$,
and the last inequality holds since $\sum_{i=1}^{k} \frac{1}{i^2} \leq \sum_{i=1}^{\infty} \frac{1}{i^2}  = O(1)$.

Thus, applying Jensen's Inequality to Inequality \eqref{eq_h2}, we have that
\begin{equation*}
\mathbb{E}[\| \hat{V}_k \hat{V}_k^\top -  V_k  V_k^\top \|_F] \leq  O\left(\frac{\sqrt{d}}{(\sigma_k- \sigma_{k+1})}      \frac{\log^{\frac{1}{2}}(\frac{1}{\delta})}{\eps}\right).
\end{equation*}
\end{proof}

\section{Conclusion and Future Work}\label{sec_conclusion}

We present a new analysis of the Gaussian mechanism for a large class of symmetric matrix approximation problems, by viewing this mechanism as a Dyson Brownian motion initialized at the input matrix $M$. 
This viewpoint allows us to leverage the stochastic differential equations which govern the evolution of the eigenvalues and eigenvectors of Dyson Brownian motion to obtain new utility bounds for the Gaussian mechanism.
To obtain our utility bounds, we show that the gaps $\Delta_{ij}(t)$ in the eigenvalues of the Dyson Brownian motion stay at least as large as the initial gap sizes (up to a constant factor), as long as the initial gaps in the top $k+1$ eigenvalues of the input matrix are $\geq \Omega(\sqrt{d})$ (Assumption \ref{assumption_gaps}).

While we observe that our assumption on the top-$k+1$ eigenvalue gaps holds on multiple real-world datasets, in practice one may need to apply differentially private matrix approximation on any matrix where the ``effective rank'' of the matrix is $k$— that is, on any matrix where the $k$’th eigenvalue gap $\sigma_k - \sigma_{k+1}$ is large— including on matrices where the gaps in the other eigenvalues may not be large and may even be zero.
 Unfortunately, for matrices with initial gaps in the top-$k$ eigenvalues smaller than $O(\sqrt{d})$, the gaps $\Delta_{ij}(t)$ in the eigenvalues of the Dyson Brownian motion become small enough that the expectation of the (inverse) second-moment term $\frac{1}{\Delta^2_{ij}(t)}$ appearing in the It\^o integral (Lemma \ref{Lemma_integral}) in our analysis may be very large or even infinite. 
 Thus, the main question that remains open is whether one can obtain similar bounds on the utility for differentially private matrix approximation for any initial matrix $M$ where the $k$’th gap $\sigma_k - \sigma_{k+1}$ is large, without any assumption on the gaps between the other eigenvalues of $M$.

Finally, this paper analyzes a mechanism in differential privacy, which has many implications for preserving sensitive information of individuals.  Thus, we believe our work will have positive societal impacts and do not foresee any negative impacts on society.

\section*{Acknowledgements}
This research was  supported in part by  NSF CCF-2104528 and CCF-2112665 awards.

\bibliography{DP}
\bibliographystyle{plain}

\newpage

\appendix

\section{Eigenvalue Gaps of Wishart Matrices} \label{appendix_wishart}

In this section, we provide the results of numerical simulations where we compute the minimum eigenvalue gap, $\min_{i\in [d-1]} \sigma_i - \sigma_{i+1}$ of Wishart random matrices $W = A^\top A$, where $A$ is an $m \times d$ matrix with i.i.d. $N(0,1)$ Gaussian entries, for various values of $d$ and $m$.
The goal of these simulations is to evaluate for what values of $m,d$ the eigenvalue gaps in a Wishart random matrix $W$ satisfy  Assumption \ref{assumption_gaps}($W,k,\lambda_1, \eps, \delta$), which requires that the size of the top-$k+1$ eigenvalues of the matrix $W$ satisfy $\sigma_i - \sigma_{i+1} \geq  \frac{8\sqrt{\log(\frac{1.25}{\delta})}}{\eps} \sqrt{d} + c\log^{\frac{1}{2}}(\lambda_1 k)$ for every $i \in [k]$.

We observe that, when $d$ is held constant and $m$ is increased, the minimum eigenvalue gap size $\min_{i\in [d-1]} \sigma_i - \sigma_{i+1}$ grows roughly proportional to $\sqrt{m}$ (Figure \ref{fig_Wishart_gaps}).
Moreover, if we fix $m$ to be $m=d^3$ and increase $d$, we observe that the minimum eigenvalue gap is at least as large as $\sqrt{d}$ with high probability and grows roughly proportional to $\sqrt{d}$ (Figure \ref{fig_Wishart_gaps_2}).
Thus, we expect Wishart random matrices to satisfy Assumption \ref{assumption_gaps}($W,k,\lambda_1, \eps, \delta$) with high probability as long as $m\geq \Omega(\log(\frac{1.25}{\delta}){\eps^2} \times d^3)$,   for any $k \leq d$ and $\eps, \delta>0$ where, e.g., $\lambda_1 \leq O(d)$.
In particular, we note that in the application of our main result to subspace recovery we set $\lambda_1 = 1$ (Corollary \ref{cor_subspace_recovery}), and in the application of our main result  to rank-$k$ covariance matrix approximation, we set $\lambda_1 = \sigma_1$ (Corollary \ref{cor_rank_k_covariance2}) and thus have that  $\lambda_1 \leq O(\sqrt{d})$ with high probability by the concentration bounds in Lemma \ref{lemma_concentration}.
All simulations were run on Matlab.

Finally, we note that there is a long line of work in random matrix theory which provides results about the distributions of the eigenvalues of random matrix ensembles (see e.g. \cite{pastur1967distribution, dyson1962brownian, arnold1971wigner, tracy1994level, tao2011random,  erdHos2012universality, huang2015bulk}), including for Wishart random matrices  \cite{silverstein1985smallest, vivo2007large, tao2012random}. 
For instance, results are given in \cite{tao2012random, wang2012random} for the local eigenvalue statistics of families of Wishart matrices where $\frac{d}{m} \rightarrow y$ as $d \rightarrow \infty$ for any fixed constant $0< y \leq 1$. 
In particular, these works include results that give high-probability bounds on the minimum eigenvalue gap of this class of Wishart matrices (see Theorems 16 and 18 in \cite{tao2012random}, and also Theorem 1.7 of  \cite{wang2012random} who extend results of \cite{tao2012random} to the edge of the spectrum).
Results for eigenvalue gap probabilities of other random matrix ensembles are given, e.g., in \cite{tao2011random, figalli2016universality, feng2019small, bourgade2021extreme}.
However, to the best of our knowledge, we are not aware of any bounds for the minimum eigenvalue gap of families of Wishart matrices where $\frac{m}{d}$ does not converge to a constant as $d \rightarrow \infty$. 
While it may be possible to extend the analysis given in \cite{tao2012random} to obtain bounds for the minimum eigenvalue gap for families of Wishart matrices where $\frac{m}{d}$ is polynomial in $d$, this analysis would be beyond the scope of our paper.

\begin{figure}[h]
    \centering
    \includegraphics[width=0.5\linewidth]{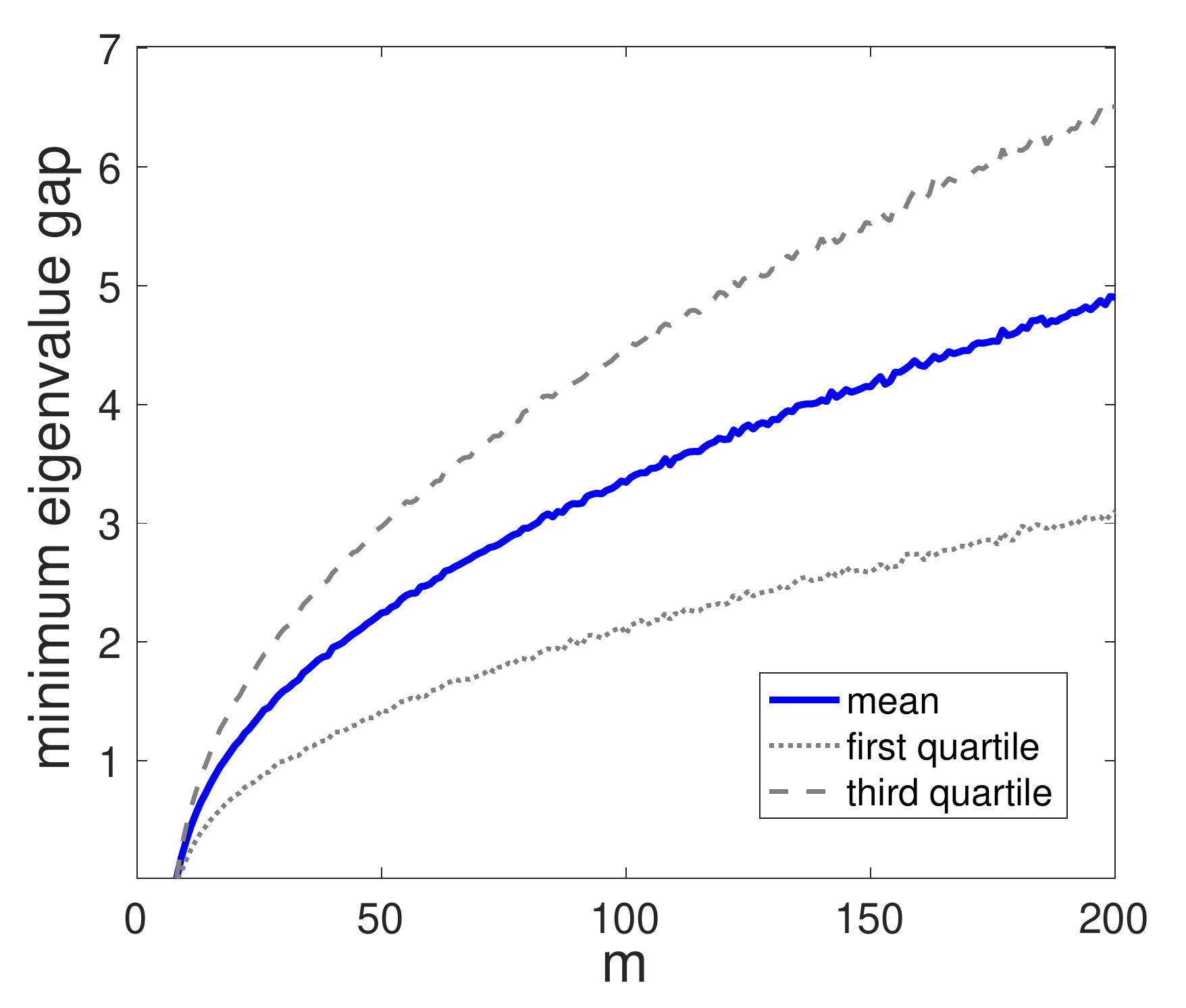}
    \caption{The minimum eigenvalue gap, $\min_{i\in [d-1]} \sigma_i - \sigma_{i+1}$ of Wishart random matrices $W = A^\top A$, where $A$ is an $m \times d$ matrix with i.i.d. $N(0,1)$ Gaussian entries, for various values of $m$ and $d=10$, averaged over 10,000 trials for each $m$ (blue curve), with first quartile (dashed grey curve) and third quartile (dotted grey curve)  also displayed.
    We observe that the minimum eigenvalue gap size grows roughly proportional to $\sqrt{m}$.}
    \label{fig_Wishart_gaps}
\end{figure}

\begin{figure}[h]
    \centering
    \includegraphics[width=0.5\linewidth]{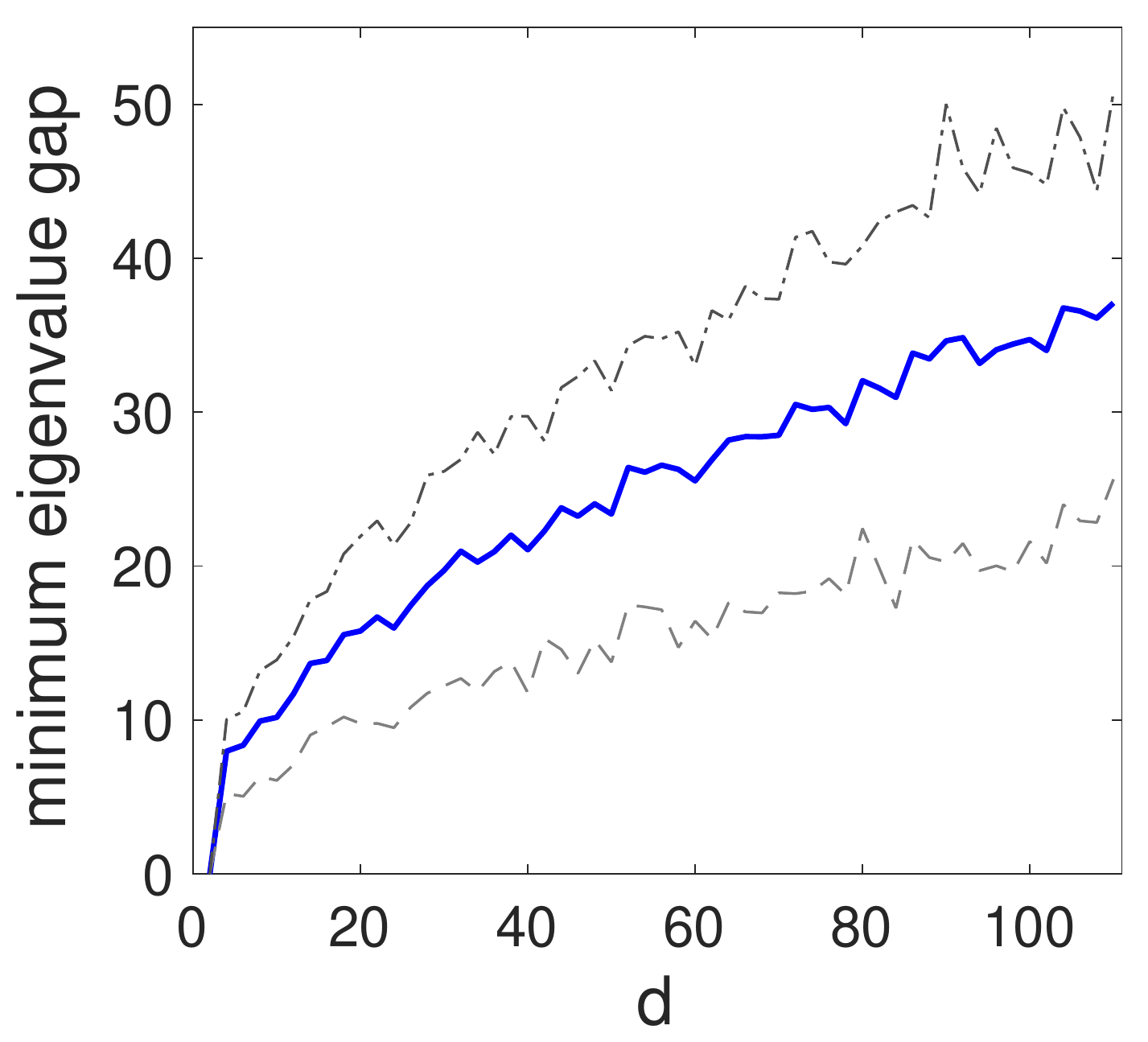}
    \caption{The minimum eigenvalue gap, $\min_{i\in [d-1]} \sigma_i - \sigma_{i+1}$ of Wishart random matrices $W = A^\top A$, where $A$ is an $m \times d$ matrix with i.i.d. $N(0,1)$ Gaussian entries, for various values of $d$ and $m=d^3$, averaged over 200 trials for each $d$ (blue curve), with first quartile (dashed grey curve) and third quartile (dotted grey curve)  also displayed.
    We observe that, if $m=d^3$, the minimum gap size is at least as large as  $\sqrt{d}$ with high probability, and grows roughly proportional to $\sqrt{d}$.}
    \label{fig_Wishart_gaps_2}
\end{figure}

\section{Eigenvalue Gaps in Real Datasets} \label{appendix_data}

In this section we compute the eigenvalues of covariance matrices $M$ of standard real-world datasets, and determine the values of $k$ for which our Assumption \ref{assumption_gaps}($M,k,\lambda_1, \eps, \delta$) holds on these datasets  (for $\lambda_1=\sigma_1$, $\epsilon=1$, and $\delta= \frac{1}{100}$) .
We consider three standard datsets from the UCI Machine Learning Repository \cite{Dua_2019}: the 1990 US Census dataset ($d=124$, $n= 2458285$), the KDD Cup dataset ($d=36$, $n = 494020$), and the Adult dataset ($d=6$, $n= 48842$). 
All three of these datasets were previously used as benchmarks in the differentially private matrix approximation and PCA literature (Census and KDD Cup in e.g. \cite{chaudhuri2012near}, and Adult in e.g. \cite{amin2019differentially}).
As is standard, we pre-process each dataset to ensure that all entries are real-valued and to normalize the range of the measurements used for the different features.
Specifically, we remove categorical features and apply min-max normalization to the remaining real-valued features.
We then multiply the data matrix by a constant to ensure that all its rows have magnitude at most $1$ (to ensure the sensitivity bounds which imply privacy of the Gaussian mechanism \cite{dwork2014algorithmic} hold on the dataset), and subtract the mean of each column.
We then compute the eigenvalues of the covariance matrix  $M = A^\top A$ of the pre-processed data matrix $A$.

Our Assumption \ref{assumption_gaps}($M,k,\lambda_1, \eps, \delta$) requires that $\sigma_i - \sigma_{i+1} \geq  \frac{8\sqrt{\log(\frac{1.25}{\delta})}}{\eps} \sqrt{d} + 3\log^{\frac{1}{2}}(\lambda_1 k)$ for all $i \leq k$.
As is done in prior works, we consider values of $k$ such that the (non-private) rank-$k$ approximation of $M$ has Frobenius norm which is a large percentage of $\|M\|_F$ (e.g., in \cite{chaudhuri2012near} they choose values of $k$ such that the Frobenius norm of the low-rank approximation is at least 80\% of the original matrix).
We will verify that, on the above-mentioned datasets, our Assumption \ref{assumption_gaps} holds for values of $k$ large enough such that the (non-private) rank-$k$ approximation contains at least $99\%$ of the Frobenius norm of $M$.

We first compute the eigenvalues of the covariance matrix of the Census dataset (Figure \ref{fig_Census}, left). 
On this dataset, the Frobenius norm of the (non-private) rank-$k$ approximation for $k =10$ contains $>99\%$ of the Frobenius norm of $M$; thus we would like our assumption to hold for values of $k$ for $k\leq 10$.
For the census dataset we have $\sigma_1 = 93730$, and thus, for any $k \leq 11$, the r.h.s. of Assumption \ref{assumption_gaps} is at most $442$.
We observe that the eigenvalue gaps satisfy $\sigma_{i} - \sigma_{i+1} \geq 442$ for all $i \leq 11$  (Figure \ref{fig_Census}, right); thus, our Assumption \ref{assumption_gaps} is satisfied for all $k \leq 11$.

Computing the eigenvalues of the covariance matrix of the KDD Cup dataset (Figure \ref{fig_KDDCUP}, left) we observe that the Frobenius norm of the (non-private) rank-$k$ approximation for $k =3$ contains $>99\%$ of the Frobenius norm of $M$; thus we would like our assumption to hold for values of $k$ at least $3$.
On this dataset we have $\sigma_1 = 72670$, and thus, for any $k \leq 7$, the r.h.s. of Assumption \ref{assumption_gaps} is at most $250$.
We observe that the eigenvalue gaps satisfy $\sigma_{i} - \sigma_{i+1} \geq 250$ for $i \leq 7$  (Figure \ref{fig_KDDCUP}, right); thus, Assumption \ref{assumption_gaps} is satisfied for all $k \leq 7$ on this dataset.

Computing the eigenvalues of the Adult dataset (Figure \ref{fig_Adult}, left) we observe that the Frobenius norm of the (non-private) rank-$k$ approximation for $k =4$ contains $99\%$ of the Frobenius norm of $M$.
On this dataset, we have that $\sigma_1 = 1195$, and thus, for any $k \leq 4$, the r.h.s. of Assumption \ref{assumption_gaps} is at most $103.4$.
We observe that the eigenvalue gaps satisfy $\sigma_{i} - \sigma_{i+1} \geq 103.4$ for $i \leq 4$  (Figure \ref{fig_Adult}, right); thus, Assumption \ref{assumption_gaps} is satisfied for all $k \leq 4$ on this dataset.

\begin{figure}[h]
    \centering
    \includegraphics[width=0.5\linewidth]{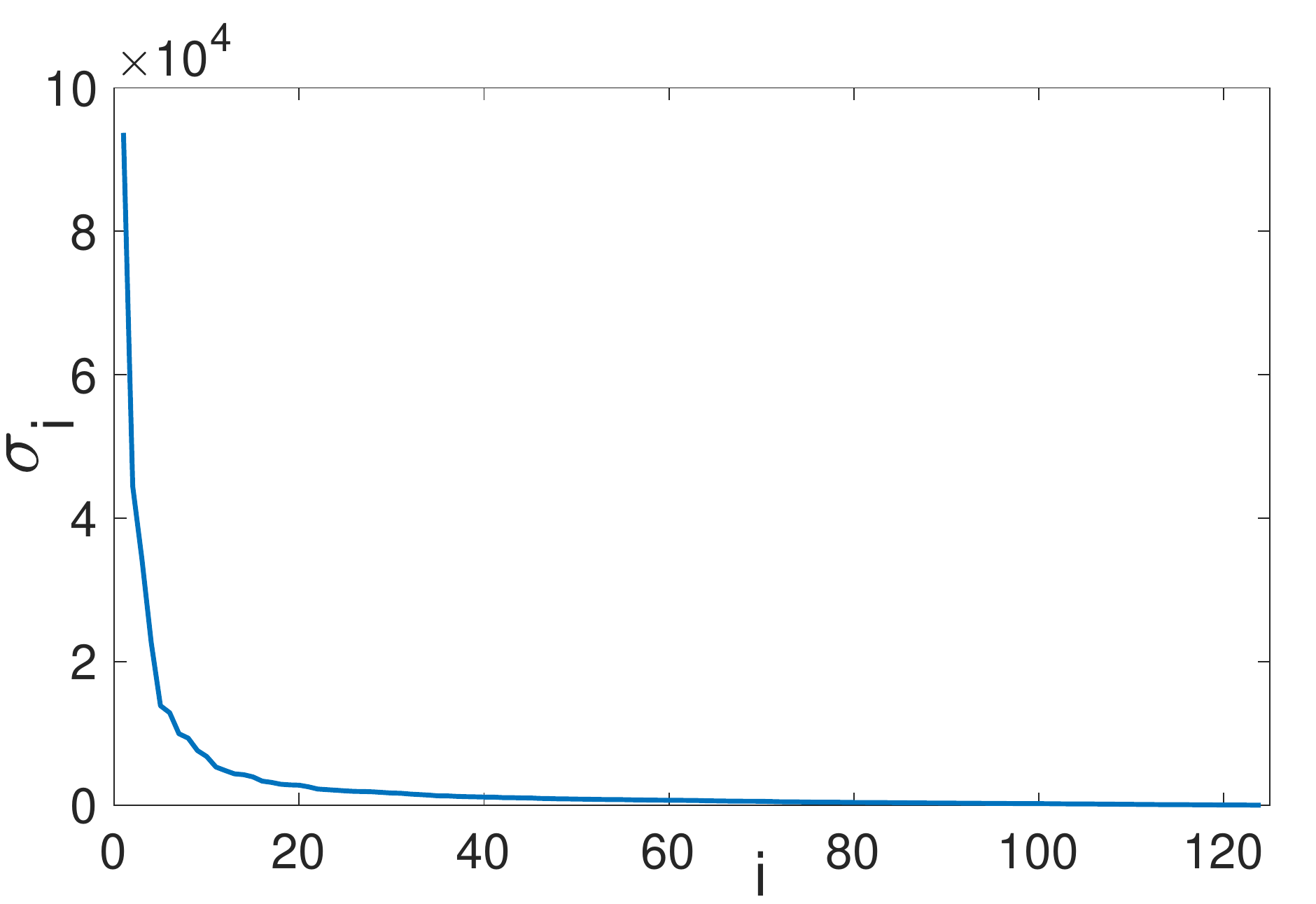}
    \includegraphics[width=0.45\linewidth]{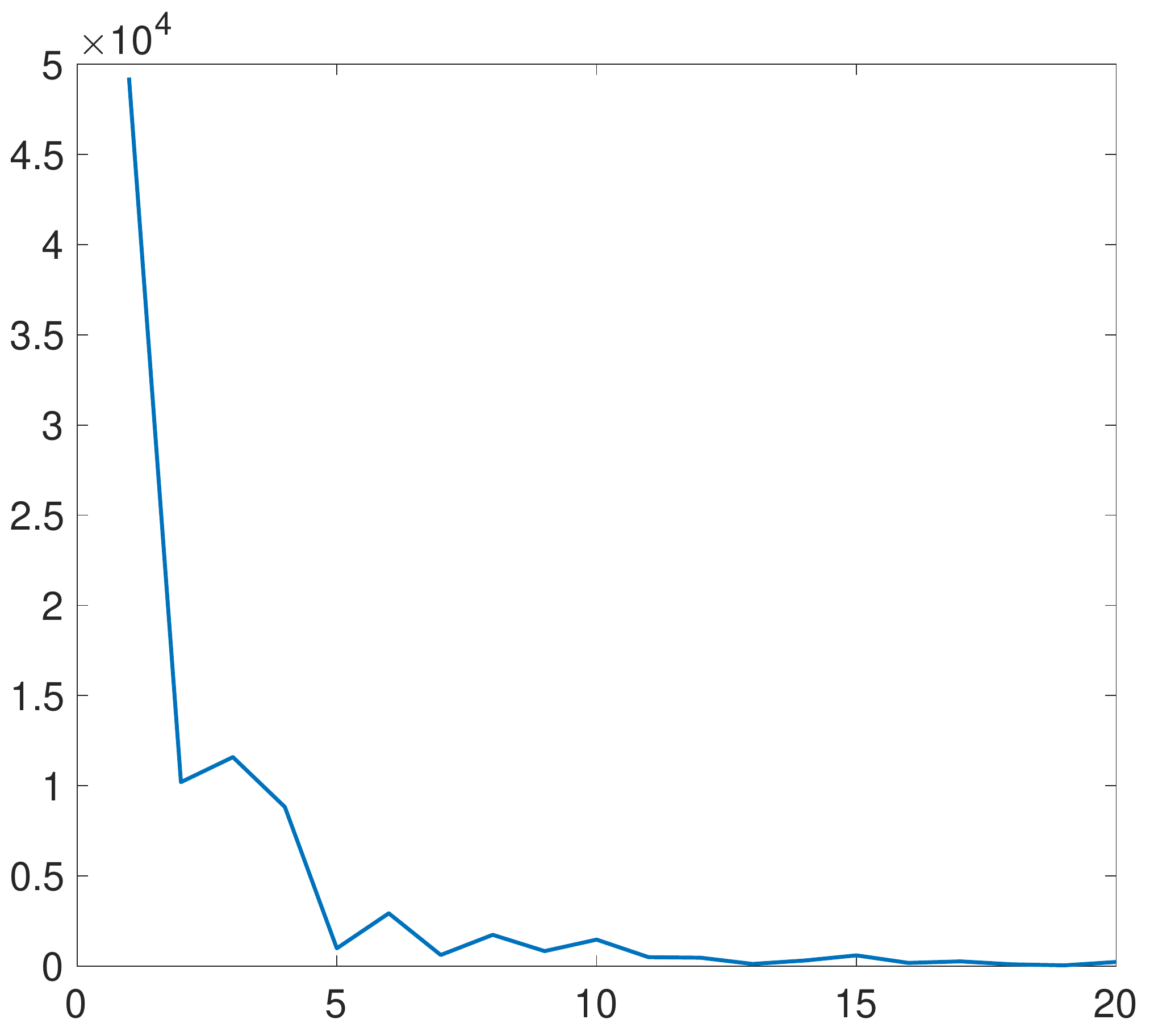}
    \caption{Eigenvalues (left) and eigenvalue gaps (right) of the covariance matrix of the Census dataset.  The gaps in the eigenvalues satisfy Assumption \ref{assumption_gaps} for any $k \leq 11$, and $\lambda_1=\sigma_1$, $\epsilon = 1$, $\delta = \frac{1}{100}$, as  for these values Assumption \ref{assumption_gaps} requires that $\sigma_i- \sigma_{i+1} \geq 442$ for all $i \leq k$.}
    \label{fig_Census}
\end{figure}

\begin{figure}[h]
    \centering
    \includegraphics[width=0.45\linewidth]{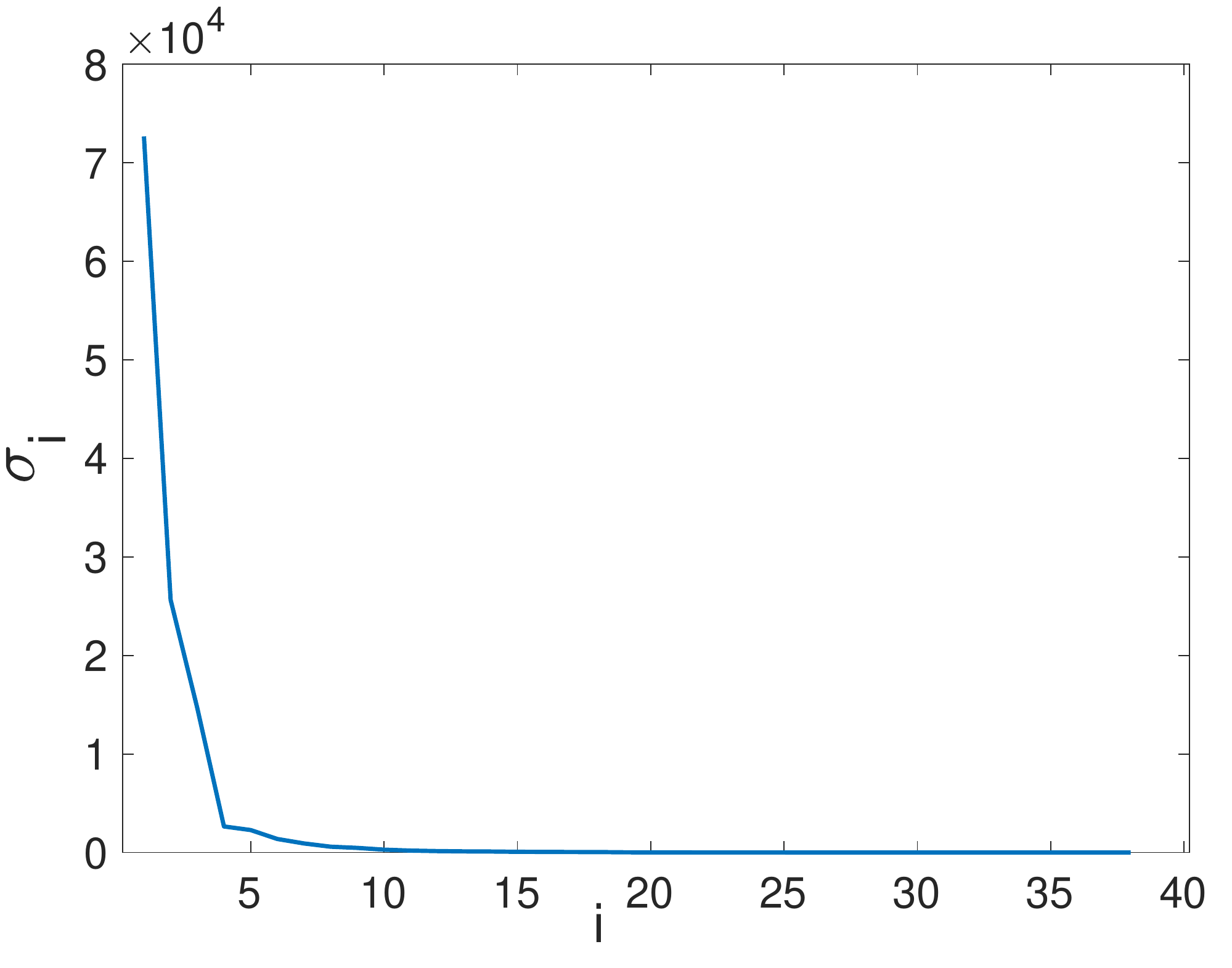}
    \includegraphics[width=0.5\linewidth]{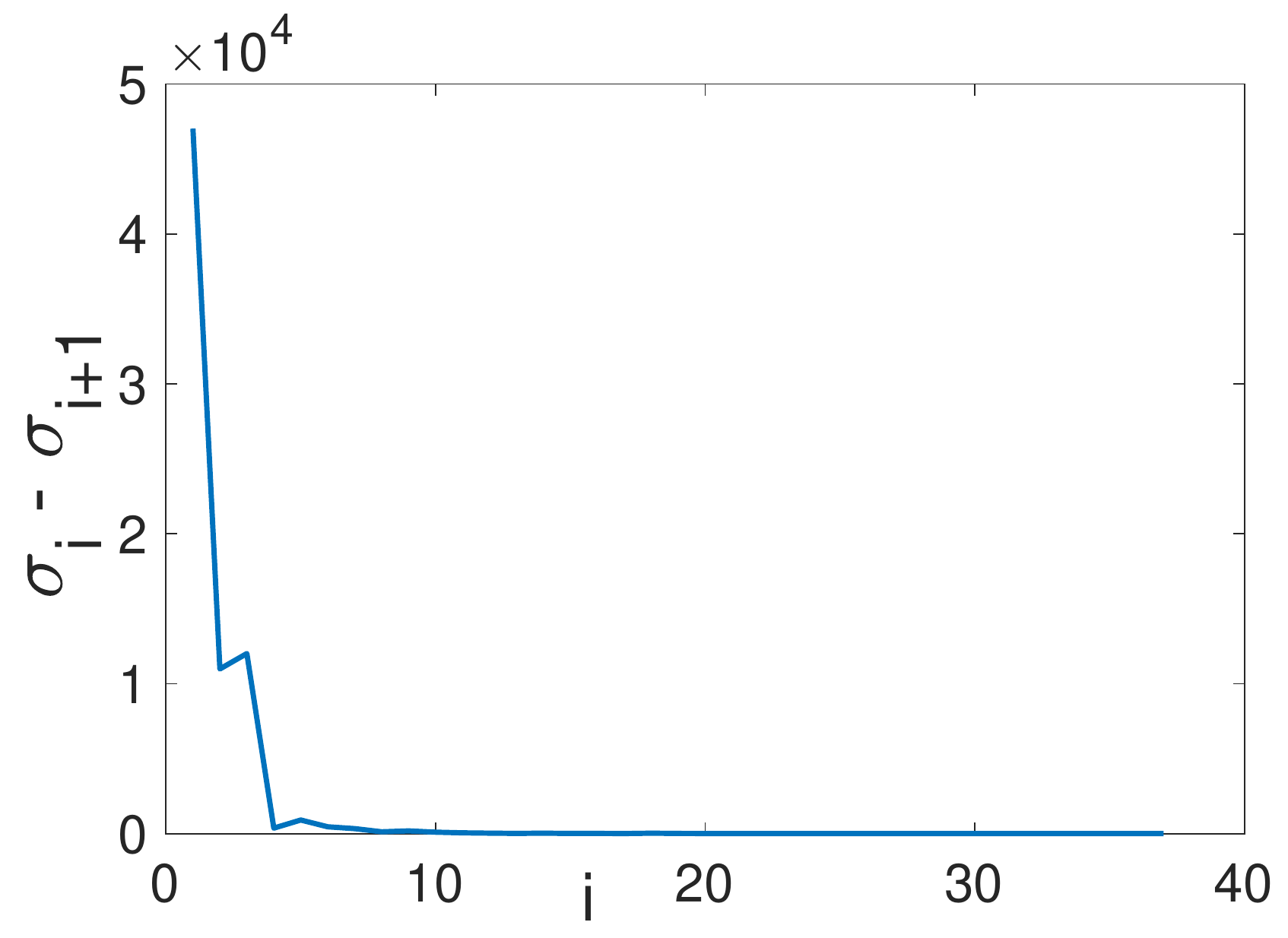}
    \caption{Eigenvalues (left) and eigenvalue gaps (right) of the covariance matrix of the KDD Cup dataset.  The gaps in the eigenvalues satisfy Assumption \ref{assumption_gaps} for any $k \leq 7$, and $\lambda_1=\sigma_1$, $\epsilon = 1$, $\delta = \frac{1}{100}$, as  for these values Assumption \ref{assumption_gaps} is satisfied when $\sigma_i- \sigma_{i+1} \geq 250$ for all $i \leq k$.}
    \label{fig_KDDCUP}
\end{figure}

\begin{figure}[h]
    \centering
    \includegraphics[width=0.45\linewidth]{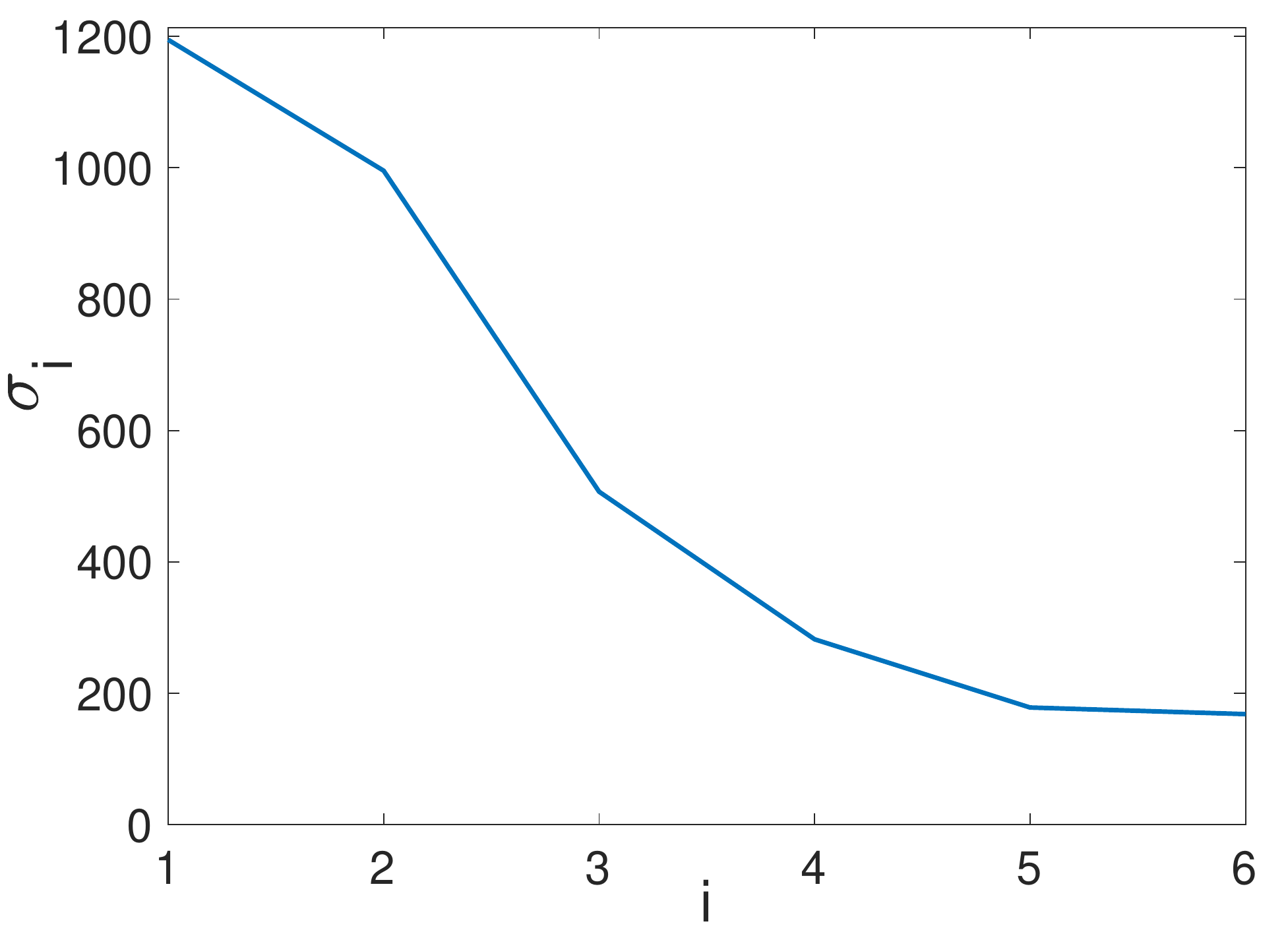}
    \includegraphics[width=0.5\linewidth]{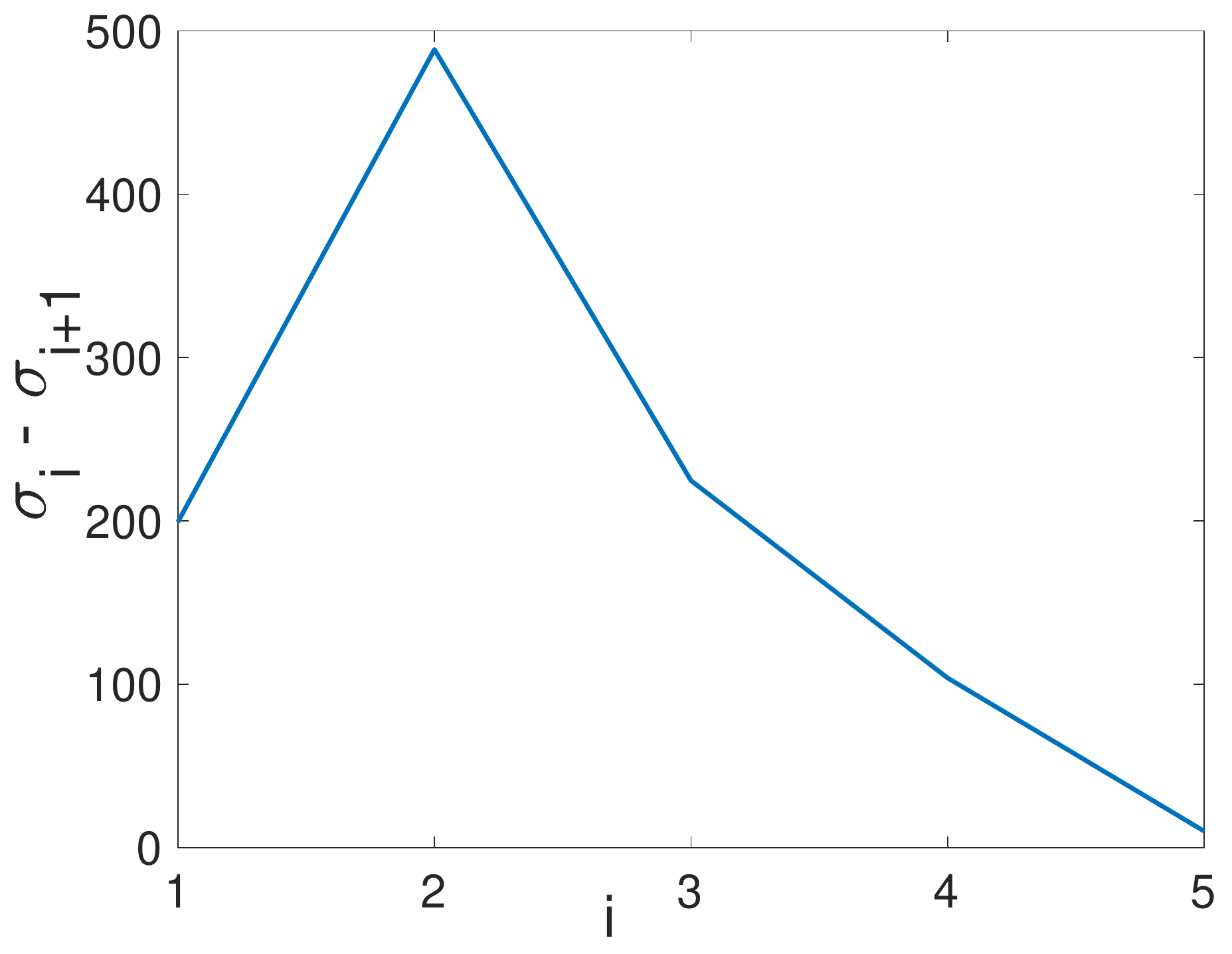}
    \caption{Eigenvalues (left) and eigenvalue gaps (right) of the covariance matrix of the Adult dataset.  The gaps in the eigenvalues satisfy Assumption \ref{assumption_gaps} for any $k \leq 4$, and $\lambda_1=\sigma_1$, $\epsilon = 1$, $\delta = \frac{1}{100}$, as  for these values Assumption \ref{assumption_gaps} is satisfied when $\sigma_i- \sigma_{i+1} \geq 103.4$ for all $i \leq k$.}
    \label{fig_Adult}
\end{figure}

\newpage
\section{Challenges in Using Previous Approaches}\label{sec_challenges}

In the special case of covariance matrix approximation, it is possible to use trace inequalities to bound the quantity $\|M- \hat{V}_k \hat{\Sigma}_k \hat{V}_k^\top\|_F - \|M- V_k \Sigma_k V_k^\top\|_F$ (which is bounded above by the quantity $\|\hat{V} \hat{\Sigma}_k \hat{V}^\top - V \Sigma_k V^\top\|_F$ we bound).
This is the approach taken in \cite{dwork2014analyze}, which applies the fact that 
\begin{equation} \label{eq_trace}
\mathrm{tr}(X) \leq \mathrm{rank}(X)\|X\|_2 \qquad \qquad \forall X \in \mathbb{R}^{d \times d}
\end{equation}
 to show that $$\|M- \hat{V}_k \hat{\Sigma}_k \hat{V}_k^\top\|_F^2 - \|M- V_k \Sigma_k V_k^\top\|_F^2 \leq O(k \|M - V_k \Sigma_k V_k^\top\|_2 \|E\|_2 + k \|E\|_2^2).$$
The r.h.s. depends on $\sigma_{k+1} = \|M - V_k \Sigma_k V_k^\top\|_2$, and is therefore not invariant to scalar multiplications of $M$.
However, one can obtain a scalar-invariant bound on the quantity $\|M- \hat{V}_k \hat{\Sigma}_k \hat{V}_k^\top\|_F - \|M- V_k \Sigma_k V_k^\top\|_F$ by plugging in $\|M - V_k \Sigma_k V_k^\top\|_2 \leq \|M - V_k \Sigma_k V_k^\top\|_F$ and plugging in the high-probability bound  $\|E\|_2 = O(\sqrt{d})$.
This leads to a bound of $\|M- \hat{V}_k \hat{\Sigma}_k \hat{V}_k^\top\|_F - \|M- V_k \Sigma_k V_k^\top\|_F \leq O(k \sqrt{d})$.
In the special case where $k=d$, this bound is $O(d^{1.5})$, and thus is not tight since we have $\|M- \hat{V}_k \hat{\Sigma}_k \hat{V}_k^\top\|_F - \|M- V_k \Sigma_k V_k^\top\|_F = \|\hat{M} - M \|_F = \|E\|_F = O(d)$ w.h.p.
Roughly, the additional factor of $\sqrt{k} = \sqrt{d}$ incurred in their bound is due to the fact that the matrix trace inequality \eqref{eq_trace} their analysis relies on gives a bound in terms of the spectral norm, even though they only need a bound in terms of the Frobenius norm-- which can (in the worst case) be larger than the spectral norm by a factor of $\sqrt{k}$.

Another issue is that the quantity $\|\hat{V} \Sigma_k \hat{V}^\top - V \Sigma_k V^\top\|_F$ we wish to bound can be very sensitive to perturbations to $V_k$, 
since $\|\hat{V} \Lambda \hat{V}^\top - V \Lambda V^\top\|_F \geq \lambda_k \|\hat{V}_k \hat{V}_k^\top - V_k V_k^\top\|_F$. 
Thus, any bound on $\|\hat{V} \Lambda \hat{V}^\top - V \Lambda V^\top\|_F$ must (at the very least) also bound the distance $\|\hat{V}_k \hat{V}_k^\top - V_k V_k^\top\|_F$ between the projection matrices onto the subspace $\mathcal{V}_k$ spanned by the top-$k$ eigenvectors of $M$.
One approach to bounding  $\|\hat{V}_k \hat{V}_k^\top - V_k V_k^\top\|_F$, is to use an eigenvector perturbation theorem, such as the Davis-Kahan theorem \cite{davis1970rotation}, which says, roughly, that
\begin{equation} \label{eq_Davis_Kahan}
    \|\hat{V}_k \hat{V}_k^\top - V_k V_k^\top\|_2 \leq \frac{\|E\|_2}{\sigma_k - \sigma_{k+1}}
\end{equation}
 (this is the approach taken by \cite{dwork2014analyze}  when proving their utility bounds for subspace recovery).
 Plugging in the high-probability bound  $\|E\|_2 = O(\sqrt{d})$, and using the fact that $\|\hat{V}_k \hat{V}_k^\top - V_k V_k^\top\|_F \leq \sqrt{k} \|\hat{V}_k \hat{V}_k^\top - V_k V_k^\top\|_2$, gives $\|\hat{V}_k \hat{V}_k^\top - V_k V_k^\top\|_F \leq \frac{\sqrt{k} \sqrt{d}}{\sigma_k - \sigma_{k+1}}$ with high probability.
To obtain bounds for the utility $\|\hat{V} \Sigma_k \hat{V}^\top - V \Sigma_k V^\top\|_F$ of the covariance matrix approximation, we can decompose $V \Sigma_k V^\top = \sum_{i=1}^k (\sigma_i- \sigma_{i+1})V_i V_i^\top$, and apply the Davis-Kahan theorem to each projection matrix $V_i V_i^\top$:
\begin{align}
    \|\hat{V} \Sigma_k \hat{V}^\top - V \Sigma_k V^\top\|_F &=  \|\hat{V} \Sigma_k \hat{V}^\top - V \Sigma_k V^\top\|_F \\
&=    \|\sum_{i=1}^{k-1} (\sigma_i - \sigma_{i+1}) (\hat{V}_i \hat{V}_i^\top - V_i V_i^\top) + \sigma_k (\hat{V}_k \hat{V}_k^\top - V_k V_k^\top)\|_F \\
    &\leq \sum_{i=1}^{k-1} (\sigma_i - \sigma_{i+1}) \|\hat{V}_i \hat{V}_i^\top - V_i V_i^\top\|_F + \sigma_k \|\hat{V}_k \hat{V}_k^\top - V_k V_k^\top\|_F \\
    & \leq \sum_{i=1}^{k-1} (\sigma_i - \sigma_{i+1}) \frac{\sqrt{i} \sqrt{d}}{\sigma_i - \sigma_{i+1}} + \sigma_k   \frac{\sqrt{k} \sqrt{d}}{\sigma_k - \sigma_{k+1}}\\
    & = O\left(k^{1.5} \sqrt{d} + \frac{\sigma_k}{\sigma_k - \sigma_{k+1}} \sqrt{k} \sqrt{d}\right),
\end{align}
where we define $\sigma_{d+1}:= \sigma_d$.  
Unfortunately, this bound is not tight up to a factor of $k$, at least in the special case where $k=d$.
As a first step to obtaining a tighter bound, we would ideally like to add up the Frobenius norm of the summands $(\sigma_i - \sigma_{i+1}) (\hat{V}_i \hat{V}_i^\top - V_i V_i^\top)$ as a sum-of-squares rather than as a simple sum, in order to decrease the r.h.s. by a factor of $\sqrt{k}$.
However, to do so we would need to bound the cross-terms $\mathrm{tr} \left((\hat{V}_i \hat{V}_i^\top - V_i V_i^\top)   (\hat{V}_j \hat{V}_j^\top - V_j V_j^\top)\right)$ for $i \neq j$.
To bound each of these cross-terms we need to carefully track the interactions between the eigenvectors in the subspaces $\mathcal{V}_i$ and $\mathcal{V}_j$ as the noise $E$ is added to the input matrix $M$.

We handle such cross-terms by viewing the addition of noise as a continuous-time matrix diffusion $\Psi(t) = M + B(t)$, whose eigenvalues $\gamma_i(t)$ and eigenvectors $u_i(t)$, $i\in [d]$, evolve over time. This allows us to ``add up'' contributions of  different eigenvectors to the Frobenius distance as a stochastic integral, 
$$||\hat{V}\Sigma_k \hat{V}^\top-V \Sigma_k V^\top||_F^2 = ||\int_0^{T}\sum_{i=1}^{d}\sum_{j \neq i} |\sigma_i-\sigma_j|\frac{\mathrm{d}B_{ij}(t)}{\gamma_i(t)-\gamma_j(t)}(u_i(t) u_j^\top(t)+u_j(t)u_i^\top(t))||_F^2,$$
where, roughly, each differential cross term  $\frac{\mathrm{d}B_{ij}(t)}{\gamma_i(t)-\gamma_j(t)}(u_i(t)u_j^\top(t)+u_j(t)u_i^\top(t))$ adds noise to the matrix $V\Sigma_k V^\top$ independently of the other terms since the Brownian motion differentials $\mathrm{d}B_{ij}(t)$ are independent for all $i,j,t$. \eqref{eq_int_1}. Roughly, this allows us to add up the contributions of these terms as a sum of squares using It\^o's Lemma from stochastic calculus \eqref{eq_int_2}.

\section{Necessity of Assumption \ref{assumption_gaps} in Our Proof}\label{sec_necessity_of_assumption}

For simplicity, assume that $\epsilon = \delta = O(1)$.
Our proof uses Weyl’s inequality to bound the gaps in the eigenvalues $\gamma_i(t) - \gamma_{i+1}(t)$ of the perturbed matrix $M + G(t) + G^\top(t)$ at every time $t \in [0,1]$, where $G(t)$ has i.i.d.  $N(0,t)$ entries.

Weyl’s inequality says that  for every $i \in d$,    $\sigma_i - \|G(t) + G^\top(t)\|_2  \leq \gamma_i(t) \leq  \sigma_i + \|G(t) + G^\top(t)\|_2$.
Thus, plugging in the high-probability bound $\|G(t)\|_2 \leq 2\sqrt{d}$,   we have that  
\begin{equation}\label{eq_gaps}
    \gamma_i(t) - \gamma_{i+1}(t) \geq \sigma_i -  \sigma_{i+1}  - 4\sqrt{d}.
    \end{equation}
If $\sigma_i - \sigma_{i+1} < \sqrt{d}$, Weyl's inequality does not give {\em any} bound on the gaps since then the r.h.s. of \eqref{eq_gaps} is $\sigma_i -  \sigma_{i+1}  - 4\sqrt{d}$ is negative.
Thus, to apply Weyl's inequality to bound the eigenvalue gaps $\gamma_i(t) - \gamma_{i+1}(t)$, we require that $\sigma_i - \sigma_{i+1} \geq \Omega(\sqrt{d})$, which is roughly  Assumption \ref{assumption_gaps}.

On the other hand, we note that Weyl's inequality is a worst-case deterministic inequality-- it says that $\sigma_i - \|G(t) + G^\top(t)\|_2  \leq \gamma_i(t) \leq  \sigma_i + \|G(t) + G^\top(t)\|_2$ with probability $1$.
However, $G(t)$ is a random matrix, and the Dyson Brownian motion equations \eqref{eq_DBM_eigenvalues} which govern the evolution of the eigenvalues of the perturbed matrix $M+ G(t)+G^\top(t)$ say that the eigenvalues $\gamma_i(t)$ and $\gamma_j(t)$ repel each other with a ``force'' of magnitude $\frac{1}{\gamma_i(t) - \gamma_j(t)}$.
This suggests that it may be possible to weaken (or eliminate) Assumption \ref{assumption_gaps}, while still recovering the same bound in our main result Theorem \ref{thm_large_gap}.

\section{High Probability Bounds}\label{sec_high_probability}

While our current result holds in expectation, it is possible to use our techniques to prove high-probability bounds. 

The simplest approach is to plug in the expectation bound in our main result (Theorem \ref{thm_large_gap}) into Markov's inequality, which says that $P(\| \hat{V} \Lambda \hat{V}^\top -  V \Lambda V^\top \|_F^2 \geq s) \leq \frac{E(\| \hat{V} \Lambda \hat{V}^\top -  V \Lambda V^\top \|_F^2)}{s}$ for all $s>0$.

While Markov's inequality gives a high-probability bound, this bound decays as $\frac{1}{s}$.
One approach to obtaining high-probability bounds which decay with a rate exponential in $s$ might be to apply concentration inequalities to the part of our proof where we currently use expectations.
Namely, our proof of Lemma \ref{Lemma_integral}  uses It\^o's Lemma (restated as Lemma \ref{lemma_ito_lemma_new}) to show that, roughly,
 $$\|\Psi(T) - \Psi(0)\|_F^2 = 4\int_0^T \sum_{i=1}^{d}  \sum_{j \neq i}  \frac{(\lambda_i - \lambda_j)^2}{(\gamma_i(t) - \gamma_j(t))^2} + \frac{1}{2} \int_0^t \sum_{\ell, r} \sum_{\alpha, \beta} \left(\frac{\partial}{ \partial X_{\alpha \beta}} f(X(t))\right) R_{(\ell r) (\alpha \beta)}(t) \mathrm{d}B_{\ell r}(t),$$
 where we define $f(X) := \|X\|_F^2$, $X(t):=  \int_0^{t}\sum_{i=1}^{d} \sum_{j \neq i} |\lambda_i - \lambda_j| \frac{\mathrm{d}B_{ij}(s)}{|\gamma_i(s) - \gamma_j(s) |}(u_i(s) u_j^\top(s) + u_j(s) u_i^\top(s))$, and  $R_{(\ell r) (i j)}(t) := \left(\frac{ |\lambda_i - \lambda_j| }{|\gamma_i(t) - \gamma_j(t)|}(u_i(t) u_j^\top(t) + u_j(t) u_i^\top(t)) \right)[\ell, r]$.
 
 The two integrals on the r.h.s. are both random variables.
 For simplicity, our current proof bounds these random variables by taking the expectation of both sides of the equation.
 In particular, the second term on the r.h.s. has  mean $0$, and thus vanishes when we apply the expectation.

We do not have to do any additional work to bound the first term on the r.h.s. with high probability, since the only random variables appearing in that term are the eigenvalue gaps $\gamma_i(t) - \gamma_j(t)$,  and we have already shown a high-probability bound on these gaps (Lemma \ref{lemma_gap_concentration}).
To bound the second term on the r.h.s. with high probability, in addition to the gaps $\gamma_i(t) - \gamma_j(t)$, we would also need to deal with the Gaussian random variables $B_{\ell r}(t)$ appearing inside the integral.
One approach to bounding these random variables $B_{\ell r}(t)$ with high probability might be to  apply standard Gaussian concentration inequalities, and it would be interesting to see whether this leads to high-probability bounds which are similar to the expectation bounds we obtain.

\end{document}